\documentclass[12pt]{article}

\usepackage[T1]{fontenc}
\usepackage{lmodern, microtype}
\usepackage[left=1.0in, right=1.0in, top=1.3in, bottom=1.3in]{geometry}
\usepackage[small]{titlesec}
\usepackage{mathtools}

\usepackage{amssymb}
\usepackage{amsmath}
\usepackage{amsthm}
\usepackage{fancyhdr}
\usepackage{bbm}
\usepackage{xcolor}
\usepackage{hyperref}
\hypersetup{
  colorlinks=true,
  linkcolor=blue,
  citecolor=red
}
\usepackage{natbib}
\setcitestyle{authoryear,open={(},close={)}}
\usepackage{asymptote}
\usepackage{subfig}

\hypersetup{
  colorlinks=true,
  linkcolor=blue,
  citecolor=red
}

\newtheorem{theorem}{Theorem}
\newtheorem{definition}{Definition}
\newtheorem*{theorem*}{Theorem}
\newtheorem{lemma}{Lemma}

\newtheorem{corollary}{Corollary}

\newtheorem*{lemma*}{Lemma}

\newcommand{\laurent}[1]{{\textcolor{blue}{{#1} ---Laurent}}}

\newtheorem{innercustomprop}{Proposition}
\newenvironment{customprop}[1]
  {\renewcommand\theinnercustomprop{#1}\innercustomprop}
  {\endinnercustomprop}

\def\maj{\text{m}}

\def\Aut{\mathrm{Aut}}

\title{Equitable Voting Rules}
\author{Laurent Bartholdi\footnote{Institute of Advanced Studies, Lyon, e-mail: laurent.bartholdi@gmail.com}, Wade Hann-Caruthers\footnote{California Institute of Technology, e-mail: whanncar@gmail.com}, Maya Josyula\footnote{California Institute of Technology, e-mail: mjosyula@caltech.edu}, \\ Omer Tamuz\footnote{California Institute of Technology, e-mail: tamuz@caltech.edu}, and Leeat Yariv\footnote{Princeton University, e-mail: lyariv@princeton.edu} \footnote{We thank Wolfgang Pesendorfer for useful comments. Tamuz gratefully acknowledges financial support from the Simons Foundation, through grant
419427. Yariv gratefully acknowledges financial support from the NSF, through grant SES-1629613. }}

\begin{document}
\maketitle
\begin{abstract}
May's Theorem \citeyearpar{may1952set}, a celebrated result in social choice, provides the foundation for majority rule. May's crucial assumption of symmetry, often thought of as a procedural equity requirement, is violated by many choice procedures that grant voters identical roles. We show that a weakening of May's symmetry assumption allows for a far richer set of rules that still treat voters equally. We show that such rules can  have minimal winning coalitions comprising a vanishing fraction of the population, but not less than the square root of the population size. Methodologically, we introduce techniques from group theory and illustrate their usefulness for the analysis of social choice questions.
\end{abstract}

\vspace{5mm}

\section{Introduction}


Literally translated to ``power of the people'', democracy is commonly associated with two fundamental tenets: equity among individuals and responsiveness to their choices. May's celebrated theorem provides foundation for voting systems satisfying these two restrictions \citep{may1952set}. Focusing on two-candidate elections, May illustrated that majority rule is unique among voting rules that treat candidates identically and guarantee symmetry and responsiveness.  

Extensions of May's original results are bountiful.\footnote{See, e.g., \cite{cantillon2002}, \cite{fey2004}, \cite{goodwin2006}, and references therein.} However, what we view as a procedural equity restriction in his original treatment---often termed anonymity or symmetry---has remained largely unquestioned.\footnote{An exception is \cite{packel1980transitive}, who relaxes the symmetry restriction and adds two additional restrictions to generate a different characterization of majority rule than May's.} This restriction requires that no two individuals can affect the collective outcome by swapping their votes. Motivated by various real-world voting systems, this paper focuses on a particular weakening of this restriction. While still capturing the idea that no voter carries a special role, our equity notion allows for a large spectrum of voting rules, some of which are used in practice, and some of which we introduce. We analyze winning coalitions of such equitable voting rules and show that they can comprise a vanishing fraction of the population, but not less than the square root of its size. Methodologically, we demonstrate how techniques from group theory can be useful for the analysis of fundamental questions in social choice. 

To illustrate our motivation, consider the stylized example of a {\em representative democracy} rule: $m$ {\em counties} each have $k$ residents. Each county selects, using majority rule, one of two representatives. Then, again using majority rule, the $m$ representatives select one of two policies (see Figure~\ref{fig:college3} for the case $m=k=3$). 


This rule does not satisfy May's original symmetry restriction: individuals could swap their votes and change the outcome. In Figure~\ref{fig:college3}, for example, suppose voters $\{1, 2, 3, 4, 5\}$ vote for representatives supporting policy A, while voters $\{6, 7, 8, 9\}$ vote for representatives supporting policy B. With the original votes, policy A would win; but swapping voters 5 and 9 would cause policy B to win.

\begin{figure}
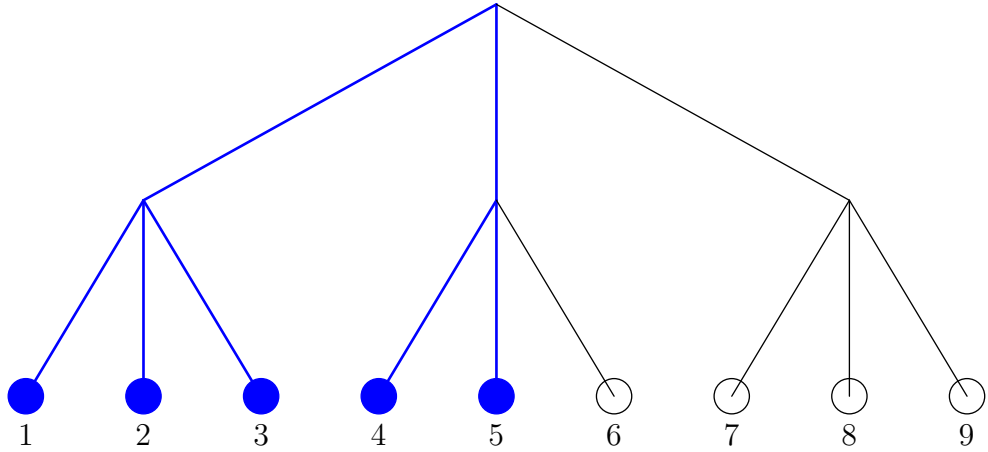
  
\begin{center}
\begin{asy}
size(13cm, 0); 
real yscale = 5.0; 
real circlesize = 0.45; 

pair ctest = (13, 5); 
real dist = 9.0; 

pair left = ctest - (dist, yscale);    	   
pair right = ctest + (dist, -yscale);     	
pair newc = ctest + (0, -yscale);

draw(ctest -- left, blue+linewidth(1)); 
// draw(ctest -- left);
draw(ctest -- newc, blue+linewidth(1));
// draw(ctest -- newc); 
// draw(ctest -- right, blue+linewidth(1));
draw(ctest -- right); 

pair l_left = left - (dist/3, yscale);
pair l_right = left + (dist/3, -yscale);
pair r_left = right - (dist/3, yscale);
pair r_right = right + (dist/3, -yscale);

draw(circle(l_left, circlesize), blue); 
draw(circle(l_right, circlesize), blue);
fill(circle(l_left, circlesize), blue); 
fill(circle(l_right, circlesize), blue);
draw(left -- l_left, blue+linewidth(1)); 
draw(left -- l_right, blue+linewidth(1)); 
// draw(circle(l_left, circlesize)); 
// draw(left -- l_left); 
// draw(circle(l_right, circlesize)); 
// draw(left -- l_right); 
label("1", l_left, 3S);
label("3", l_right, 3S); 

pair l_center = left + (0, -yscale); 
// draw(circle(l_center, circlesize)); 
draw(circle(l_center, circlesize), blue);
fill(circle(l_center, circlesize), blue);
// draw(left -- l_center); 
draw(left -- l_center, blue+linewidth(1));
label("2", l_center, 3S);

// draw(circle(r_left, circlesize), blue); 
// draw(circle(r_right, circlesize), blue);
// fill(circle(r_left, circlesize), blue); 
// fill(circle(r_right, circlesize), blue);
// draw(right -- r_left, blue+linewidth(1)); 
// draw(right -- r_right, blue+linewidth(1)); 
draw(circle(r_left, circlesize)); 
draw(right -- r_left); 
draw(circle(r_right, circlesize)); 
draw(right -- r_right); 
label("7", r_left, 3S);
label("9", r_right, 3S); 

pair r_center = right + (0, -yscale); 
draw(circle(r_center, circlesize)); 
draw(right -- r_center); 
label("8", r_center, 3S);

pair c_left = newc - (dist/3, yscale);
// draw(circle(c_left, circlesize));   
draw(circle(c_left, circlesize), blue);
fill(circle(c_left, circlesize), blue);
pair c_right = newc + (dist/3, -yscale); 
draw(circle(c_right, circlesize)); 
pair c_center = newc + (0, -yscale); 
// draw(circle(c_center, circlesize)); 
draw(circle(c_center, circlesize), blue);
fill(circle(c_center, circlesize), blue);
// draw(newc -- c_left); 
draw(newc -- c_left, blue+linewidth(1)); 
// draw(newc -- c_center); 
draw(newc -- c_center, blue+linewidth(1));
draw(newc -- c_right); 
label("4", c_left, 3S); 
label("5", c_center, 3S); 
label("6", c_right, 3S); 

\end{asy}
\end{center}
\caption{A representative democracy voting rule. Voters are grouped into three counties: $\{1,2,3\}$, $\{4,5,6\}$ and $\{7,8,9\}$. Each county elects a representative by majority rule, and the election is decided by majority rule of the representatives.}\label{fig:college3}
\end{figure}

Even though representative democracy rules do not satisfy May's symmetry assumption, there certainly is an intuitive sense in which their fundamental characteristics ``appear'' equitable. Indeed, variations of these rules were chosen in good faith by many designers of modern democracies. Such rules are currently in use in France, India, the United Kingdom, and the United States, among others. 

Is there a formal sense in which a representative democracy rule is more equitable than a dictatorship? More generally, what makes a voting rule equitable? We suggest the following definition. In an {\em equitable voting rule}, for any two voters $v$ and $w$, there is some permutation of the full set of voters such that: (i) the permutation sends $v$ to $w$, and (ii) applying this permutation to any voting profile leaves the election result unchanged. 

In \S\ref{sec:roles}, we formalize the notion of ``roles'' in a voting body. For instance, in university committees there are often two distinct roles: a chair and a standard committee member; likewise, in juries, there is often a foreperson, who carries a special role, and several jury members, who all have the same role. We show that our notion of equity is tantamount to \textit{all} agents in the electorate having the same role.


Under our equity definition, representative democracy rules are indeed equitable, but dictatorships are not. For instance, in the case depicted in Figure~\ref{fig:college3}, voters $1$ and $2$ play the same role, since the permutation that swaps them leaves any election result unchanged. But $1$ and $4$ also play the same role: the permutation that swaps the first \textit{county} with the second \textit{county} also leaves outcomes unchanged.

There is a large variety of equitable rules that are not representative democracy rules. An example is what we call {\em Cross Committee Consensus (CCC)} rules. In these, each voter is assigned to two committees: a ``row committee'' and a ``column committee'' (see Figure~\ref{fig:ccc}). If any row committee and any column committee both exhibit consensus, then their choice is adopted. Otherwise, majority rule is followed. For instance, suppose a university is divided into equally-sized departments, and each faculty member sits on one university-wide committee. CCC corresponds to a policy being accepted if there is a strong unanimous lobby from a department and from a university-wide committee, with majority rule governing decisions otherwise. This rule is equitable since each voter is a member of precisely one committee of each type, and all row (column) committees are interchangeable.

\begin{figure}
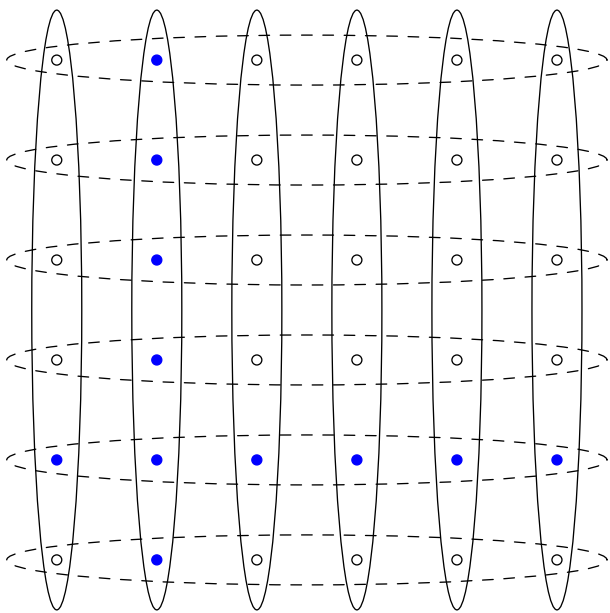

\begin{center}
\begin{asy}
size(8cm,0);
int n = 6;
real center = (n - 1) / 2; 
real width = n / 2; 
real circlesize = 0.05;  

for (int i = 0; i < n ; ++i)
{
	for (int j = 0; j < n ; ++j)
    {
    	if (i == 0)
        { 
        	draw(ellipse((center, j), width, 0.25), black+dashed);
        }
    	if (i == 1 || j == 1)
        {
        	draw(circle((i, j), circlesize), blue); 
        	fill(circle((i, j), circlesize), blue);
        } else {
    		draw(circle((i, j), circlesize));
        }
    }
    draw(ellipse((i, center), 0.25, width)); 
}
\end{asy}
\end{center}
\caption{Cross-committee consensus voting rule. The union of a row and a column is a winning coalition.}\label{fig:ccc}
\end{figure}

\vspace{3mm}

We provide a number of further examples of equitable voting rules, showing the richness of this class and its versatility in allowing different segments of  society---counties, university departments, etc.---to express their preferences. 


In order to characterize more generally the class of equitable voting rules, we focus on their {\em winning coalitions}, the sets of voters that decide the election when in agreement \citep{reiker1962}. In majority rule, all winning coalitions include at least half of the population. We analyze how small winning coalitions can be in equitable voting rules. 

Under representative democracy, winning coalitions have to comprise at least a quarter of the population.\footnote{A winning coalition under representative democracy must include support from half the counties, which translates to half of the population in those counties, or one quarter of the entire population.} Much smaller winning coalitions are possible in what we call {\em generalized representative democracy (GRD)} voting rules, where voters are hierarchically divided into sets that are, in turn, divided into subsets, and so on. For each set, the outcome is given by majority rule over the decisions of the subsets.\footnote{These rules have been studied under the name {\em recursive majority} in the probability literature \citep[see, e.g., ][]{mossel1998recursive}.}
We show that equitable GRD rules for $n$ voters can have winning coalitions as small as $n^{\log_3 2}$, or about $n^{0.63}$, which is a vanishingly small fraction of the population.\footnote{For an example of a non-equitable GRD with a small winning coalition, consider voters $\{1,\ldots,1000\}$ and assume three counties divide the population into three sets of voters: $\{\{1\},\{2\},\{3,\ldots,1000\}$. Then $\{1,2\}$ is a winning coalition. The value $\log_3 2 \approx 0.63$ is the Hausdorff dimension of the Cantor set. As it turns out, there is a connection between equitable GRD's that achieve the smallest winning coalitions and the Cantor set.}


When the number of voters $n$ is a perfect square, and when committee sizes are taken to be $\sqrt{n}$, the CCC rule has a winning coalition of size $2\sqrt{n}-1$. This is significantly smaller than $n^{\log_3 2}$, for $n$ large enough.

Our main result is that, for any $n$, there always exist simple equitable voting rules that have winning coalitions of size $\approx 2\sqrt{n}$. Conversely, we show that no equitable voting rule can have winning coalitions of size less than $\sqrt{n}$. Methodologically, the proof utilizes techniques from group theory and suggests the potential usefulness of such tools for the analysis of collective choice.


While $\sqrt{n}$ accounts for a vanishing fraction of the voter population, we stress that it can be viewed as ``large'' in many contexts. While in a department of $100$ faculty, $10$ members would need to coordinate to sway a decision one way or the other, in a presidential election with, say, $140$ million voters, coordination between nearly $12,000$ voters would be necessary to impact outcomes.\footnote{Interestingly, rules that give decisive power to minorities of size $\sqrt{n}$ appear in other contexts of collective choice and have been proposed for apportioning representation in the United Nations Parliamentary Assembly, and for voting in the Council of the European Union, see \cite{zyczkowski2014square}. These proposed rules relied on the Penrose Method \citep{penrose1946}, which suggests the vote weight of any representative should be the square root of the size of the population she represents, when majority rule governs decisions. Penrose argued that this rule assures equal voting powers among individuals.} 

For instance, even under majority rule, if each voter is equally likely to vote for either of two alternatives, the Central Limit Theorem suggests that a coalition of order $\sqrt{n}$ can control the vote with high probability.


Certainly, beyond equity, another important aspect of voting rules is their susceptibility to manipulation. For instance, with information on voters' preferences, representative democracy rules are sensitive to gerrymandering \citep{gerrymandering2016}. We view the question of manipulability as distinct from that of equity. It would be interesting to formulate a notion of non-manipulability, independent of equity, and to understand how these notions interact. The breadth of equitable voting rules allows for further consideration of various objectives, such as non-manipulability, when designing institutions.

In the next part of the paper, we explore a stronger notion of equity. We consider {\em $k$-equitable} voting rules in which every coalition of $k$ voters plays the same role. These are increasingly stringent conditions that interpolate between our equity notion, when $k=1$, and May's symmetry, when $k=n$. The analysis of $k$-equitable rules is delicate, due to group- and number-theoretical phenomena. There do exist, for arbitrarily large population sizes $n$, voting rules that are $2$- and $3$-equitable, and have winning coalitions as small as $\sqrt{n}$. However, for ``most'' sufficiently large values of $n$, and for any $k \geq 2$, the {\em only} $k$-equitable, neutral, and responsive voting rule is majority. Thus, while equity across individuals allows for a broad spectrum of voting rules, equity among arbitrary fixed-size coalitions usually places the restrictions May had suggested. While $k$-equity is arguably a strong restriction, it is still far weaker than May's original symmetry requirement. In that respect, our results here provide a strengthening of May's conclusions.


\section{The Model}
\label{sec:model}
\subsection{Voting rules}
Let $V$ be a finite set of voters. We denote $V=\{1,\ldots,n\}$ so that $n$ is the number of voters. Each voter has preferences over alternatives in the set $Y = \{-1, 1\}$. We identify the possible preferences over $Y$ with elements of $X = \{ -1, 0, 1 \}$, where $-1$ represents a strict preference for $-1$ over $1$, $1$ represents a strict preference for $1$ over $-1$, and $0$ represents indifference between $-1$ and $1$. We denote by $\Phi = X^V$ the set of voting profiles; that is, $\Phi$ is the set of all functions from the set of voters $V$ to the set of possible preferences $X$. A \textit{voting rule} is a function $f \colon \Phi \rightarrow X$. 

\vspace{3mm}

An important example is the \textit{majority} voting rule $\maj\colon \Phi \to X$, which is given by
  \begin{align*}
    \maj(\phi)= 
    \begin{cases}
        1 &\text{if } |\phi^{-1}(1)| > |\phi^{-1}(-1)|\\
        -1 &\text{if } |\phi^{-1}(1)| < |\phi^{-1}(-1)|\\
        0 &\text{otherwise}.
    \end{cases}
  \end{align*}
A vote of $0$ can be interpreted as abstention or indifference. 


\subsection{May's Theorem}

We now define several properties of voting rules.
Following \cite{may1952set}, we say that a voting rule $f$ is \textit{neutral} if $f(-\phi)=-f(\phi)$. Neutrality implies that both alternatives $-1$ and $1$ are treated symmetrically: if each individual flips her vote, the final outcome is also flipped.

Again following \cite{may1952set}, we say that a voting rule $f$ is \textit{positively responsive} if increased support for one alternative makes it more likely to be selected. Formally, $f$ is positively responsive if $f(\phi)=1$ whenever there exists a voting profile $\phi'$ satisfying the following:
\begin{enumerate}
    \item $f(\phi') = 0$ or $1$.
    \item $\phi(v) \geq \phi'(v)$ for all $v \in V$.
    \item $\phi(v_0) > \phi'(v_0)$ for some $v_0 \in V$.
\end{enumerate}
Thus, $f(\phi)\ge f(\phi')$ if $\phi\ge\phi'$ coordinate-wise, and if $f(\phi)=0$ then any change of $\phi$ breaks the tie.


We now turn to the symmetry between voters. Several group-theoretic concepts will prove useful for the description and comparison of May's and our notions. Denote by $S_n$ the set of permutations of the $n$ voters. Any permutation of the voters can be associated with a permutation of the set of voting profiles $\Phi$: given a permutation $\sigma \in S_n$, the associated permutation on the voting profiles maps $\phi$ to $\phi^\sigma$, which is given by $\phi^\sigma(v) = \phi (\sigma^{-1} v)$.  The \textit{automorphism group} of the voting rule $f$ is given by
\begin{align*}
  \Aut_f =  \{ \sigma \in S_n \, | \, \forall \phi \in \Phi, \, f(\phi^\sigma) = f(\phi)  \}.
\end{align*}
That is, $\Aut_f$ is the set of permutations of the voters that leave election results unchanged, \textit{for every} voting profile.

We can interpret a permutation $\sigma$ as a scheme in which each voter $v$, instead of casting her own vote, gets to decide how some {\em other voter} $w=\sigma(v)$ will vote. A permutation $\sigma$ is in $\Aut_f$ if applying this scheme never changes the outcome: when each $w = \sigma(v)$ votes as $v$ would have, the result is the same as when each player $v$ votes for herself.

The automorphism group $\Aut_f$ has natural implications for pivotality, or the Shapley-Shubik and Banzhaf indices of players in simple games, see \cite{dubey1979mathematical}. Consider a setting in which all voters choose their votes identically and independently at random. Given such a distribution, we can consider the probability $\eta_v$ that a voter $v$ is pivotal.\footnote{A voter $v$ is pivotal at a particular voting profile if a change in her vote can affect the outcome under $f$.} It is easy to see that if there is some $\sigma \in \Aut_f$ that maps $v$ to $w$, then $\eta_v=\eta_w$, implying that $v$ and $w$ have the same Banzhaf index. In fact, when there exists $\sigma \in \Aut_f$ that maps $v$ to $w$, any statistic associated with a voter that treats other voters identically---the probability the outcome coincides with voter $v$'s vote, the probability that voter $v$ and another voter are pivotal, etc.---would be the same for voters $v$ and $w$. Hence, in an equitable rule, all the voters' Banzhaf indices will be equal.\footnote{In a recent follow-up paper to this paper, \cite{bhatnagar2020voting} shows that the converse does not hold: there are rules that are not equitable, but for which the same holds.}

\cite{may1952set}'s notion of equity, often termed \textit{symmetry} or \textit{anonymity}, requires that swapping the votes of any two voters will not affect the collective outcome. It can be succinctly stated as $\Aut_f = S_n$.  

\vspace{3mm}

\noindent\textbf{May's Theorem.} \textit{
Majority rule is the unique symmetric, neutral, and positively responsive voting rule.} 

\vspace{3mm}

Perhaps surprisingly, the requirement of symmetry is stronger than what is needed for May's conclusions. As it turns out, a weaker requirement that $\Aut_f$ be restricted to only {\em even permutations} would suffice for his results, see Lemma~\ref{lem:alternating-or-symmetric-then-big-winning-coalitions} in the appendix.\footnote{A permutation $\sigma$ is even if the number of pairs $(v,w)$ such that $v<w$ and $\sigma(v)>\sigma(w)$ is even. Put another way, define a transposition to be a permutation that only switches two elements, leaving the rest unchanged. A permutation is even if it is the composition of an even number of transpositions.} In Theorem~\ref{thm:6-transitive} below we show that, in fact, a far weaker requirement suffices. 

\subsection{Equitable Voting Rules}

As we have already seen, the requirement that $\Aut_f$ includes all permutations, or all even permutations, precludes many examples of voting rules that ``appear'' equitable. What makes a voting rule appear equitable? Our view is that, in an equitable voting rule, ex-ante, all voters carry the same ``role.'' We propose the following definition and discuss in the next section the sense in which it formalizes this view. 

\begin{definition}
A voting rule $f$ is \textit{equitable} if for every $v,w \in V$ there is a $\sigma \in \Aut_f$ such that $\sigma(v)=w$.
\end{definition}

In words, a voting rule is equitable if, for any two voters $v$ and $w$, there is some permutation of the population that relabels $v$ as $w$ such that, \textit{regardless of voters' preferences}, the outcome is unchanged relative to the original voter labeling. 


In group-theoretic terms, $f$ is equitable if and only if the group $\Aut_f$ \textit{acts transitively} on the voters.\footnote{It turns out our equity restriction is effectively the definition of transitivity. The notion of transitive groups is not directly related to transitivity of relations often considered in Economics.} Insights from group theory related to the characteristics of transitive groups are therefore at the heart of our main results. Appendix~\ref{sec:primer} contains a short primer on the basic group theoretical background that is needed for our analysis. \footnote{\cite{isbell1960homogeneous} studied a notion equivalent to our equity notion and considered its implications, combined with rule neutrality, in a setting in which preferences must be strict, so that tie breaking must also be equitable. When $n$ is odd, majority is equitable and neutral. Isbell showed that for some even $n$ such rules exist, while for other even $n$ they do not.}



\subsection{Equity as Role Equivalence}
\label{sec:roles}

May noted the strong link between anonymity and equality, stating that ``This condition might well be termed anonymity... A more usual label is equality'' \citep[page 681]{may1952set}. Are anonymity and equality inherently one and the same? In this section, we formalize this question in terms of {\em roles}. This allows for a natural distinction between May's symmetry or anonymity condition and our equity notion. In particular, we formalize our motivating idea that equity corresponds to all voters carrying the same role.



Even though we defined voting rules with respect to a given set of voters, the design of voting rules is often carried out without a particular group of people in mind; rather, collective institutions are often designed in the abstract. For example, hiring protocols in university departments might be specified prior to any specific search. One such protocol might be that a committee chair decides dictatorially whom to hire, unless indifferent, in which case the committee decides by majority rule. Under such an abstract rule, the dictatorial privilege is not assigned to a particular Prof.\ X, but to an abstract role called ``chair.'' Later, when a committee is formed, the role of chair is assigned to some particular faculty member. The same applies, at least in aspiration, to countries' election rules, jury decision protocols, etc. Indeed, historical cases in which laws were written with particular individuals specified or implied are often not benevolent examples of institution building.

To capture this idea, we introduce abstract voting rules. Recall that we define (non-abstract) voting rules as functions from  $X^V$ to $X$ in the context of a particular voter set $V$. An {\em abstract} voting rule is a map $f \colon X^R \to X$ for a set of {\em roles} $R$. Of course, mathematically, these objects are identical, and so we can speak of abstract voting rules as being anonymous, equitable, etc. In the above example of the committee, the set of roles would be $R=\{C,M_2,\ldots,M_n\}$, where $C$ stands for ``chair'' and $M_i$ is member $i$.

The conceptual difference between voters and roles is that voters have preferences and vote, whereas roles do not. Therefore, for a vote to take place, voters need to be assigned to roles. Accordingly, given a group of voters $V$ equal in size to $R$, we call a bijection $a \colon V \to R$ a {\em role assignment}. Given an abstract voting rule $f \colon X^R \to X$, a role assignment $a$ defines a (non-abstract) voting rule $f_a \colon X^V \to X$ in the obvious way, via $f_a(\phi) = f(\phi \circ a^{-1})$. In our university committee example, if $V = \{\text{Alex},\text{Bailey},\ldots\}$, an assignment $a$ that satisfies $a(\text{Alex})=C$ and $a(\text{Bailey})=M_7$ assigns Alex the role of chair, and Bailey the role of member $7$. Hence, the voting rule $f_a$ is a dictatorship of Alex. A different assignment $b$ with $b(\text{Bailey})=C$ and $b(\text{Alex})=M_7$ results in the voting rule $f_b$ in which Bailey is the dictator. Note that any assignment $c$ that also assigns $c(\text{Bailey})=C$ results in the same voting rule as $b$, even if, say, $c(\text{Alex}) = M_8$: $f_b(\phi) = f_c(\phi)$ for any voting profile $\phi$. In the context of an abstract voting rule $f$, we say that two assignments are equivalent if they lead to the same voting rule: $a$ and $b$ are equivalent under $f$ if $f_a=f_b$.

The next proposition captures a sense in which symmetry is a form of equality:
\begin{customprop}{1A}
\label{clm:anonymous}
An abstract voting rule $f \colon X^R \to X$ is symmetric if and only if all assignments are equivalent under $f$.  
\end{customprop}
Thus, symmetry means that assignments do not matter, or that voters are completely indifferent between assignments: given any voting profile $\phi$ and given $f$, each voter would be indifferent if given a choice between assignments. Voters do not care which role they have; moreover, voters do not care which roles other voters have.

We now turn to the interpretation of our equity notion in terms of roles. In the university committee example, certainly ``chair'' is a distinguished role. Likewise, it is clear that if we view the representative democracy example of Figure~\ref{fig:college3} as an abstract voting rule, no role is distinguished. Of course, once we assign roles to voters with particular preferences, some voters may be disadvantaged, and prefer a different assignment ex-post. In that respect, in the abstract representative democracy rule, all roles are ex-ante identical.

We say that roles $r_1, r_2 \in R$ are equivalent under an abstract voting rule $f \colon X^R \to X$ if, for any voter $v$ and assignment $a$ such that $a(v)=r_1$, there is an assignment $b$ with $b(v)=r_2$ such that $f_a=f_b$. That is, roles $r_1$ and $r_2$ are equivalent if, for any voter, it is impossible to determine whether they are assigned the role of $r_1$ or $r_2$ from the entire mapping from vote profiles to chosen alternatives. In the university committee example, $M_i$ and $M_j$ are equivalent, but $C$ is not equivalent to any other role. If $a(v)=C$, then clearly this can be determined from $f_a$, but not if $a(v) \neq C$; in the latter case one can find an assignment $b$ such that $b(v) \neq a(v)$, but $f_a=f_b$, since it is impossible to tell whether a voter has role $M_i$ or $M_j$. In the representative democracy example of Figure~\ref{fig:college3}, it is impossible to determine from $f_a$ the role of any given voter $v$.

Given this definition of equivalent roles, the next proposition is a sharp characterization of equity.
\begin{customprop}{1B}
  \label{prop:identical}
  An abstract voting rule $f \colon X^R \to X$ is equitable if and only if all roles are equivalent under $f$.
\end{customprop}
Therefore, while symmetry means that voters are indifferent between assignments, equity implies that voters are indifferent between roles. Indeed, Proposition~\ref{prop:identical} implies that, given an abstract voting rule $f$, and given any two roles $r_1,r_2$, if a voter had to choose between (i) any assignment in which she had role $r_1$, or (ii) any assignment in which she had role $r_2$, she would be indifferent. Both would allow her to select the same voting rule $f_a$. Likewise, if the voter had to choose between roles $r_1$ and $r_2$ knowing that an adversary would get to choose the rest of the assignment, she would be indifferent.

The following highlights the idea that equity captures indifference between roles.

\begin{customprop}{2A}
  \label{clm:role-transitive}
  An abstract voting rule $f \colon X^R \to X$ is equitable iff there is a set of assignments $A$ such that
  \begin{enumerate}
    \item $f_a=f_b$ for all $a, b \in A$, and
    \item for each role $r \in R$ and voter $v \in V$ there is an $a \in A$ with $a(v)=r$.
  \end{enumerate}
\end{customprop}
Thus, there is a menu of assignments that all induce the same voting rule, but allow $v$ to choose any role.


\subsection{Winning Coalitions}

One way to describe a voting rule is through its winning coalitions: the sets of individuals whose consensual vote determines the alternative chosen. Formally, we say that a subset $S \subseteq V$ is a \textit{winning coalition} with respect to the voting rule $f$ if, for every voting profile $\phi$ and $x \in \{-1,1\}$, $\phi(v) = x$ for all $v \in S$ implies $f(\phi) = x$.


Note that no two winning coalitions of $f$ can be disjoint. Indeed, suppose that $M, M' \subseteq V$ are two disjoint winning coalitions. We can then have a profile under which members of $M$ vote unanimously for $-1$ and members of $M'$ vote for $1$. In such cases, $f$ would not be well defined. The following lemma illustrates a version of the converse.

\begin{lemma}
\label{lem:make-voting-rule-from-winning-coalitions}
Let $\mathcal{W}$ be a collection of subsets of $V$ such that every pair of subsets in $\mathcal{W}$ has a nonempty intersection. Then there is a neutral, positively responsive voting rule for $V$ for which every set in $\mathcal{W}$ is a winning coalition.
\end{lemma}

Intuitively, the construction underlying Lemma \ref{lem:make-voting-rule-from-winning-coalitions} is as follows. First, for any vote profile in which a subset $W \in \mathcal{W}$ votes for $1$ (or $-1$) in consensus, we specify the voting rule to also take the value of $1$ (or $-1$). For any profile in which no $W \in \mathcal{W}$ votes in consensus, we define the voting rule to follow majority rule. By definition, the winning coalitions of this voting rule contain the sets in $\mathcal{W}$. As we show, it is also neutral and positively responsive.

This lemma will allow us to discuss neutral and positive-responsive equitable voting rules through the restrictions they impose on winning coalitions.

\section{Winning Coalitions for Equitable Voting Rules}\label{sec:coalitions}

In this section, we provide bounds on the size of winning coalitions in general equitable voting rules. We then restrict attention to the special class of equitable voting rules that generalize representative democracy rules and characterize the size of winning coalitions for those.

\subsection{Winning Coalitions of Order $\mathbf{\sqrt{n}}$}

We first show that for any population size, there always exist equitable voting rules that have winning coalitions that are of order $\sqrt{n}$.

\begin{theorem}
\label{thm:transitive-upper}
  For every $n$ there exists a neutral, positively responsive equitable voting rule with winning coalitions of size $2\lceil \sqrt{n} \;\rceil - 1$.
\end{theorem}

An important implication of this theorem is that, under an equitable voting rule, winning coalitions can account for a vanishing fraction of the population. Nevertheless, $\sqrt{n}$ is arguably a large number of voters in some contexts. For example, in a majority vote, suppose all voters vote for each of $\{-1,1\}$ independently with probability one half. A manipulator who wants to guarantee an outcome with high probability would need to control an order of $\sqrt{n}$ of the votes. This is a consequence of the fact that the standard deviation of the number of voters who vote $1$ is of order $\sqrt{n}$.

The cross committee consensus rule  described in the introduction is an example of an equitable voting rule in which winning coalitions are $O(\sqrt{n})$. Nevertheless, there is an algebraic subtlety---the construction of that rule relies on $n$ being an integer squared. Certainly, an analogous construction can be made for any $n$ that can be described as $n = k \cdot m$ for some integers $k$ and $m$ by considering some committees to be of size $k$ and others to be of size $m$. Such constructions, however, would not necessarily generate voting rules with winning coalitions of size close to $\sqrt{n}$. We prove Theorem~\ref{thm:transitive-upper} by constructing a simple, related rule that applies to every $n$, called the {\em longest-run} rule.\footnote{We thank Elchanan Mossel for suggesting this improvement to a previous construction.}

\begin{figure}
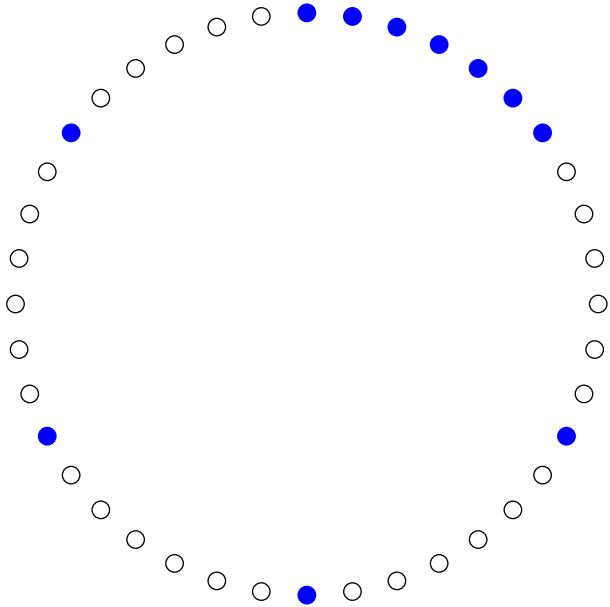

\begin{center}
\begin{asy}
int mod1 (int x, int y)
{
    int result = x - y;
    if (x < y)
    {
        return x; 
    } else {
        while (result >= y)
        {
            result = result - y; 
        }
        return result;
    }
    
}

size(8cm, 0); 
real circlesize = 0.03;

int npoints = 40; 
int m = ceil(sqrt(npoints)); 

real angle = 0; 
real interval = (2*pi)/npoints; 
for (int i = 1; i < npoints + 1;  ++i)
{
    pair loc = (cos(-angle+pi/2), sin(-angle+pi/2)); 
    if (i <= m || mod1(i, m) == 0)
    {
        draw(circle(loc, circlesize), blue); 
        fill(circle(loc, circlesize), blue);
    } else {
        draw(circle(loc, circlesize));
    }
     
    angle = angle + interval; 
}

\end{asy}
\end{center}
\caption{The longest-run voting rule. The set depicted in blue is a winning coalition of size $\approx 2 \sqrt{n}$.}
\label{fig:cyclic}
\end{figure}
Identify the set of voters with $\{0,1,\ldots,n-1\}$, and place them along a cycle, as in Figure~\ref{fig:cyclic}. Given a voting profile $\phi$, a {\em run} is a contiguous block of voters voting identically for either $-1$ or $1$. Formally, a set of voters $W \subseteq V$ is a run if $\phi(w) = \phi(w') \in \{-1,1\}$ for all $w,w' \in W$, and if $W = \{i, i+1,\ldots,i+k\}$ modulo $n$.

Given a voting profile $\phi$, we say that $W$ is the {\em longest-run} if it is a run that is strictly longer than all other runs. The longest-run voting rule $\ell$ is defined as follows: if there is a longest-run in $\phi$, then $\ell(\phi)$ is the vote cast by the members of this run. Otherwise, $\ell(\phi) = \maj(\phi)$, where $\maj$ is the majority rule.

The longest-run rule is equitable, since the group of rotations maps any voter to any voter. Furthermore, it admits winning coalitions of size $\approx 2 \sqrt{n}$: these include a run of length $\sqrt{n}$, together with $\sqrt{n}$ agents spaced $\sqrt{n}$ apart, thus preventing the creation of longer runs. See Figure~\ref{fig:cyclic}.

\medskip

We now offer a counterpart for Theorem~\ref{thm:transitive-upper} that provides a lower bound on the size of minimal coalitions in equitable voting rules.

\begin{theorem}
\label{thm:transitive-lower}
  Every winning coalition of an equitable voting rule has size at least $\sqrt{n}$.
\end{theorem}


The proof of Theorem \ref{thm:transitive-lower} relies on group-theoretic results described in Appendix~\ref{sec:primer}. To gain some intuition for the bound, suppose, as in the longest-run voting rule above, that voters are located on a circle and that $\Aut_f$ includes all rotations. These are the permutations of the form $\sigma(i) = i + k \text{ mod } n$. We know that two winning coalitions cannot be disjoint. Take, then, any winning coalition $S$ and denote by $S+k$ the winning coalition that is derived by adding $k$ (again, modulo $n$) to the label of each member. It follows that $S$ and $S+k$ must have a non-empty intersection, or that there are some $i,j \in S$ with $i-j=k$. Therefore, if we look at all the differences between two elements of $S$ (i.e., expressions of the form $i-j$, where $i,j \in S$), they encompass all $n$ rotations. In particular, the cardinality of these differences is $n$. On the other hand, the number of such differences is certainly bounded by the number of ordered pairs of members in $S$, which is $|S|^2$. It follows that $|S|^2 \geq n$, generating our bound.

\subsection{Generalized Representative Democracy Rules}

As already discussed, voting rules mimicking representative democracy are equitable, if not symmetric \`a la \cite{may1952set}. We now consider a class of equitable voting rules that generalize representative democracy rules. These capture the flavor of various hierarchical voting structures that contain more than two layers. For example, voters may belong to counties, which comprise states, which constitute a country. As we show, these sorts of hierarchical decision rules are associated with far smaller winning coalitions than $n/2$, but still substantially larger than $\sqrt{n}$.  

A voting rule $f\colon \Phi \to X$  is a {\em generalized representative democracy (GRD)} if the following hold.
  \begin{itemize}
    \item If $V=\{v\}$ is a singleton, then $f(\phi)=\phi(v)$.
    \item If $V$ is not a singleton, there exists a partition 
    $\{V_1,\dots,V_d\}$ of $V$ into $d$ sets such that
    \begin{align*}
      f(\phi) = \maj(f_1(\phi |_{V_1}), f_2(\phi |_{V_2}),\dots,f_d(\phi |_{V_d})),
    \end{align*}
    where each $f_i\colon X^{V_i}\to X$ is some generalized representative democracy rule, $\phi |_{V_i}$ is $\phi$ restricted to $V_i$, and $\maj$ is the majority rule.
  \end{itemize}
  
Any GRD rule is associated with a rooted tree that captures voters' hierarchical structure (as in Figure~\ref{fig:college} for the case of $d=3$). A GRD voting rule is equitable if, in the induced tree, the vertices in each level have the same degree.\footnote{Intuitively, the permutations required to shift one voter's role into another require the shift of that voter's entire ``county'' into the target role's ``county'', which can be done only when their numbers coincide.} 

The following result characterizes the size of winning coalitions in GRD voting rules.

\begin{theorem}
\label{thm:college}
If $f$ is an equitable generalized representative democracy rule for $n$ voters, then a winning coalition must have size at least $n^{\log_3 2}$. Conversely, for arbitrarily large $n$, there exist equitable generalized representative democracy voting rules with winning coalitions of size $n^{\log_3  2}$.
\end{theorem}

There is an intriguing connection between this characterization and the so-called Hausdorff dimension of the Cantor set, which is $\log_3 2 \approx 0.63$.\footnote{The Cantor set can be constructed by starting from, say, the unit interval and iteratively deleting the open middle third of any sub-interval remaining. That is, in the first iteration we are left with $[0,1/3]\cup[2/3,1]$, in the second iteration we are left with $[0,1/9]\cup[2/9,1/3]\cup[2/3,7/9]\cup[8/9,1]$, etc. The fractal or Hausdorff dimension is a measure of ``roughness'' of a set. See \cite{peitgen1993fractals} and references therein.} The connection arises from the fact that GRD rules with the smallest winning coalitions are those in which, at each level, the subdivision is into three groups. In such rules, to construct a winning coalition, two of the three top ``representatives'' need to agree. Then, two of the voters of these representatives need to agree, and so on recursively. This precisely mimics the classical construction of the Cantor set. 


\begin{figure}[h!]
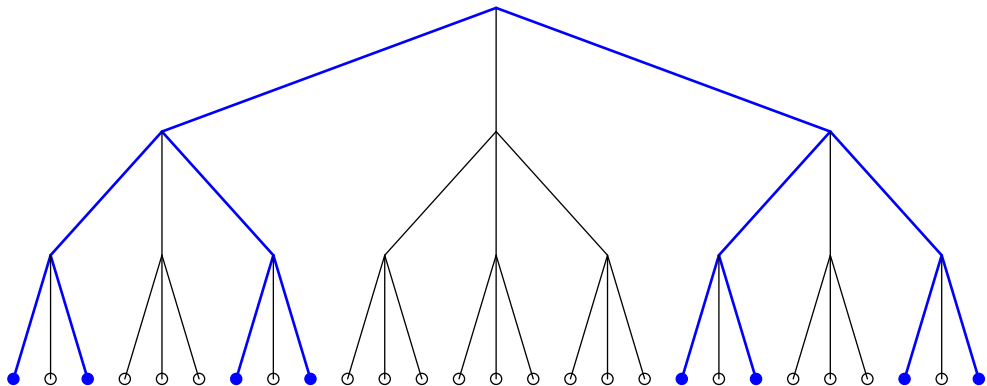

\begin{center}
\begin{asy}
size(13cm, 0); 
real yscale = 10.0; 
real circlesize = 0.45; 

void drawLR (pair center, real dist, bool isCenter)
{ 
	if (dist == 1.0)
    {
     
    }
    else if (dist == 3.0)
    {
    	pair left = center - (dist, yscale);
    	pair right = center + (dist, -yscale);
    	if (isCenter == false) {
        	draw(circle(left, circlesize), blue); 
            draw(circle(right, circlesize), blue);
            fill(circle(left, circlesize), blue); 
            fill(circle(right, circlesize), blue);
        	draw(center -- left, blue+linewidth(1)); 
        	draw(center -- right, blue+linewidth(1)); 
        }
        else {
            draw(circle(left, circlesize)); 
            draw(circle(right, circlesize));
            draw(center -- left);
            draw(center -- right);
        }
    	 
    	pair newcenter = center + (0, -yscale); 
    	draw(circle(newcenter, circlesize)); 
    	draw(center -- newcenter); 
    }
    else
    {
    	pair left = center - (dist, yscale);    	   
    	pair right = center + (dist, -yscale);     	
    	pair newcenter = center + (0, -yscale); 
   	 	
   	 	if (isCenter == false) {
   	 	    draw(center -- left, blue+linewidth(1)); 
   	 	    draw(center -- right, blue+linewidth(1));
   	 	} else {
   	 	    draw(center -- left);
   	 	    draw(center -- right); 
   	 	}
   	 	
    	draw(center -- newcenter); 
        
    	drawLR(left, dist/3, false); 
    	drawLR(newcenter, dist/3, true);
    	drawLR(right, dist/3, false); 
        
    }
}

void drawcenter(pair center, real dist)
{ 
	if (dist == 1.0)
    {
     
    }
    else if (dist == 3.0)
    {
    	pair left = center - (dist, yscale);
    	draw(circle(left, circlesize));   
    	pair right = center + (dist, -yscale); 
        draw(circle(right, circlesize)); 
    	pair newcenter = center + (0, -yscale); 
    	draw(circle(newcenter, circlesize)); 
    	draw(center -- left); 
    	draw(center -- newcenter); 
    	draw(center -- right);     
    }
    else
    {
    	pair left = center - (dist, yscale);    	   
    	pair right = center + (dist, -yscale);     	
    	pair newcenter = center + (0, -yscale); 
        
    	draw(center -- left); 
    	draw(center -- newcenter); 
    	draw(center -- right); 
        
    	drawcenter(left, dist/3); 
    	drawcenter(newcenter, dist/3);
    	drawcenter(right, dist/3); 
        
    }

}

pair ctest = (13, 5); 
real dist = 27.0; 

pair left = ctest - (dist, yscale);    	   
pair right = ctest + (dist, -yscale);     	
pair newc = ctest + (0, -yscale);

draw(ctest -- left, blue+linewidth(1)); 
draw(ctest -- newc); 
draw(ctest -- right, blue+linewidth(1));

drawLR(left, dist/3, false); 
drawLR(right, dist/3, false); 
drawcenter(newc, dist/3); 

\end{asy}
\end{center}
\caption{Generalized representative democracy voting rule. The leafs of the tree (at the bottom) represent the voters. At each intermediate node the results of the three nodes below are aggregated by majority.}\label{fig:college}
\end{figure}

\section{$\mathbf{k}$-Equitable Voting Rules}

So far, we have focused on voting rules in which individuals are indifferent between roles. Naturally, one could extend the notion and contemplate rules that are robust to larger coalitions of voters changing their roles in the population. This section analyzes such rules for arbitrary size $k$ of coalitions. With these harsher restrictions on collective-choice procedures, results similar to May's reemerge, although with important caveats. 

\begin{definition}
A voting rule is \textit{$k$-equitable} for $k \geq 1$ if, for every pair of ordered $k$-tuples $(v_1, \dots, v_k)$ and $(w_1, \dots, w_k)$ (with $v_i \neq v_j$ and $w_i \neq w_j$ for all $i \neq j$), there is a permutation $\sigma \in \Aut_f$ such that $\sigma(v_i) = w_i$ for $i = 1, \dots, k$. 
\end{definition}

\vspace{3mm}


Intuitively, $k$-equitable voting rules are ones in which every group of $k$ voters has the same ``joint role'' in the election. This restriction is certainly harsher than that imposed for equitable rules. Indeed, consider the representative democracy example of Figure~\ref{fig:college3}. Suppose Alex is assigned the role of $1$, while Bailey is assigned the role of $2$.The implication of $2$-equity is that the pair (Alex, Bailey) could potentially be associated with any pair $(i,j)$. But this is clearly not true here, since Alex and Bailey are in the same county, and thus it is impossible to associate them with, e.g., the roles of $1$ and $4$.\footnote{In group-theoretic terms, $f$ is $k$-equitable if and only if the group $\Aut_f$ is \textit{$k$-transitive}.}

In terms of roles and abstract voting rules, $k$-equity admits the following generalization of Proposition~\ref{clm:role-transitive}.
\begin{customprop}{2B}
  \label{clm:role-k-transitive}
  An abstract voting rule $f \colon X^R \to X$ is $k$-equitable iff there is a set of assignments $A$ such that
  \begin{enumerate}
    \item $f_a=f_b$ for all $a, b \in A$.
    \item For each set of $k$ roles $\{r_1,\ldots,r_k\} \subseteq R$ and $k$ voters $\{v_1,\ldots,v_k\} \subset V$ there is an $a \in A$ with $a(v_i)=r_i$ for $i \in \{1,\ldots,k\}$.
  \end{enumerate}
\end{customprop}
Thus, there is a menu of assignments that all yield the same rule, but that places no restrictions on which roles a coalition of $k$ voters takes. A perhaps illuminating analogy would be to think of equity as corresponding to strategy-proofness, and to think of $k$-equity as corresponding to strategy-proofness for coalitions of size $k$. Alternatively, $k$-equity is reminiscent of group-envy-freeness notions considered in allocation problems \citep[see, e.g.,][]{equity1974}. In a sense, $k$ is a parameter that interpolates between equity, corresponding to $k=1$, and May's symmetry, corresponding to $k=n$. 


We begin by examining $2$-equitable rules. Certainly, majority rule is $2$-equitable. As can be easily verified, none of the voting rule examples mentioned so far, other than majority, is $2$-equitable. As it turns out, for most population sizes, winning coalitions of size at least $n/2$ are endemic to $2$-equitable voting rules.

We say that almost every natural number satisfies a property $P$ if the subset $N_P \subseteq \mathbb{N}$ of the natural numbers that have property $P$ satisfies
\begin{align*}
    \lim_{n \to \infty}\frac{|N_P \cap \{1,\ldots,n\}|}{n} = 1.
\end{align*}

\begin{theorem}
\label{thm:2-eq-neg}
  For almost every natural number $n$, every $2$-equitable voting rule for $n$ voters has no winning coalitions of size less than $n/2$. In particular, for almost all $n$, the only $2$-equitable, neutral, positively responsive voting rule is majority.
\end{theorem}
Thus, for almost all $n$, the assumption of symmetry in May's Theorem can be substituted with the much weaker assumption of $2$-equity.

The proof of Theorem \ref{thm:2-eq-neg} relies on modern group-theoretical results that were not available when May's Theorem was introduced \citep{cameron1982}. As it turns out, there is a vanishing share of integers for which there exist $2$-transitive groups that are neither the set of all permutations nor the set of even permutations. We show in Lemma~\ref{lem:alternating-or-symmetric-then-big-winning-coalitions} in the appendix that those latter groups yield winning coalitions of size at least $n/2$, implying the result.

Theorem \ref{thm:2-eq-neg} states that for most population sizes, $2$-equitable rules imply large winning coalitions. This notably does not hold for \textit{all} population sizes. In Appendix~\ref{sec:2-3-equitable}, we construct $2$-equitable and $3$-equitable voting rules with small winning coalitions, which apply to arbitrarily large population sizes.

When considering more stringent equity restrictions, results are much starker and conclusions hold for \textit{all} population sizes.
\begin{theorem}
\label{thm:6-transitive}
Every $6$-equitable voting rule has no winning coalitions of size less than $n/2$. In particular, the only $6$-equitable, neutral, positively responsive voting rule is majority.\footnote{Furthermore, if $n > 24$, the same conclusions follow from the weaker assumption of $4$-equity.} 
\end{theorem}

The proof of Theorem \ref{thm:6-transitive}  relies on discoveries from the 1980's and 1990's that showed that, for any $n$, the only $6$-transitive groups are the group of all permutations and the group of even permutations. These results are a consequence of the successful completion of a large project, involving thousands of papers and hundreds of authors, called the {\em Classification of Finite Simple Groups}, see \cite{aschbacher2011classification}. 


We stress that $k$-equity is a strong restriction.  Nonetheless, for any fixed $k$, $k$-equity is a far weaker restriction than May's symmetry. In that respect, Theorem \ref{thm:6-transitive}, like Theorem~\ref{thm:2-eq-neg}, offers a strengthening of May's original result.



\section{Towards a Characterization of Equitable Voting Rules}
\label{sec:pd}

As our examples throughout illustrate, the set of equitable voting rules is broad and does not admit a simple universal procedural description. Their full characterization would be as complex as the full classification of finite simple groups alluded to above, and hence is beyond the scope of this paper. 


We start here by classifying the equitable voting rules for electorate sizes of the form $n = p$ and $n = p^2$ with $p$ a {\em prime}. We do so for two reasons. First, these cases are easier to handle, while still illustrating some of the complexities entailed in the general characterization of equitable voting rules. To gain some intuition for why these cases are easier, consider the representative democracy rule. When $n=p$, the only representative democracy rule is majority, since there is no way to divide the voters into non-singleton counties of equal size. In particular, the class of equitable rules for $n=p$ is drastically restricted. The second reason we focus on these cases is that it allows us to contemplate general voting rules that can be applied for all electorate sizes $n$.

For our classification, it will be useful to define {\em cyclic} voting rules and {\em $2$-cyclic} voting rules.

A voting rule $f$ for $n$ voters is said to be {\em cyclic} if one can identify the voters with the set $\{0,\ldots,n-1\}$ in such a way that the permutation $\sigma \colon \{0,\ldots,n-1\} \to \{0,\ldots,n-1\}$ given by $\sigma(i) = i + 1 \text{ mod } n$ is an automorphism of $f$. Intuitively, a voting rule is cyclic if the voters can be arranged on a circle in such a way that rotating, or shifting all voters one space to the right---tantamount to an application of $\sigma$---does not affect the outcome. An example of a cyclic voting rule is the longest-run rule described in \S\ref{sec:coalitions}. A perhaps less obvious example is that of the representative democracy rule; the arrangement on the cycle entails positioning members of the same county equidistantly along the cycle. In the example of Figure~\ref{fig:college} in which voters $\{1,2,\ldots,9\}$ are arranged in counties $\{\{1,2,3\},\{4,5,6,\},\{7,8,9\}\}$, if we arrange the voters along the cycle by the order $(1,4,7,2,5,8,3,6,9)$, then applying $\sigma$,  results in the order $(4,7,2,5,8,3,6,9,1)$, and so the first county is mapped to the second, the second is mapped to the third, and the third back to the first. Thus, the voting rule is unchanged by $\sigma$, and hence it is cyclic.

A voting rule $f$ for $n = n_1 \times n_2$ voters is said to be {\em $2$-cyclic} if one can identify the voters with the set $\{0,\ldots,n_1-1\} \times \{0,\ldots,n_2-1\}$ in such a way that the permutations $\sigma_1$ and $\sigma_2$ given by 
\begin{align*}
\sigma_1(i_1,i_2) &= (i_1+1\text{ mod } n_1,\,i_2)\\    
\sigma_2(i_1,i_2) &= (i_1,\,i_2+1\text{ mod } n_2)    
\end{align*}
are automorphisms of $f$. Intuitively, a voting rule is $2$-cyclic if one can arrange the voters on an $n_1 \times n_2$ grid such that shifting all voters to the right (and wrapping the rightmost ones back to the left) or shifting all voters up (again wrapping the topmost ones to the bottom) does not affect the outcome. It is easy to see that the cross committee consensus rule is an example of such a rule, as is the representative democracy rule. 




\begin{customprop}{3}
\label{prop:pd}
Let $f \colon X^V \to X$ be an equitable voting rule, and let $p$ be prime. If $n=p$, then $f$ is cyclic. If $n=p^2$, then $f$ is either cyclic, or $2$-cyclic, or both.
\end{customprop}
The generalization to $n=p^d$ for $d > 2$ is more intricate, and is not a straightforward extension to higher dimensional $d$-cyclic rules. See our discussion in \S\ref{app:pd}.


This proposition has an important implication to the design of simple, equitable voting rules that are not tailored to particular electorate sizes. Majority rule is one such rule---one only needs to tally the votes and consider the difference between the number of supporters of one alternative relative to the other. The longest-run voting rule is another such example. In fact, any such rule must also work for electorates comprised of a prime number of voters, and in particular \textit{must} be cyclic. It would be interesting to understand whether a far larger class of rules than cyclic voting rules can work for \textit{almost all} electorate sizes.

\section{Conclusions}

In this paper we study equity, a notion of procedural fairness that captures equality between different voters' roles.

The voting rules that satisfy May's symmetry axiom admit a simple description: a voting rule is symmetric, or anonymous, if and only if the outcome depends only on the number of voters who choose each possible vote. In contrast, the set of equitable voting rules is much richer, and its complexity is intimately linked to frontier problems in mathematics; in particular, the classification of finite simple groups. This paper includes a number of diverse examples, including generalized representative democracy, cross committee 
consensus, and a number of additional constructions that appear in the appendix, but the class of all equitable rules is larger yet. Understanding which equitable rules satisfy different desirable conditions---non-malleability, inclusiveness, etc.---could be an interesting avenue for future research.


We believe the approach taken here could potentially be useful for various other contexts. For example, symmetric games are often thought of as ones in which \textit{any} permutation of players' identities does not affect individual payoffs \citep[e.g.,][page 18]{dasgupta1986existence}. As is well known, such finite games have symmetric equilibria. Interestingly, in his original treatise on games, Nash took an approach to symmetry that is similar to ours, studying the automorphism group of the game. He showed that equity, analogously defined for games, suffices for the existence of symmetric equilibria: that is, it suffices that for every two players $v$ and $w$ there is an automorphism of the game that maps $v$ to $w$ \citep[page 289]{nash1951non}.\footnote{Theorem 2 in \cite{nash1951non} states that ``Any finite game has a symmetric equilibrium point.'' Now, Nash's definition of a symmetric equilibrium is a strategy profile $\mathfrak{s}$ such that $\mathfrak{s}_i = \mathfrak{s}_j$ whenever there is an automorphism $\chi$ of the game with $\chi(i)=j$. In particular, for there to be a symmetric strategy in the sense we usually consider, where all players use the same strategy, it suffices for there to be a transitive automorphism group, i.e.\ one in which for every $i,j$ there is an automorphism $\chi$ with $\chi(i)=j$.} It would be interesting to explore further the consequences of equity so defined in more general strategic interactions.

\newpage
\appendix

\section{A Primer on Finite Groups}
\label{sec:primer}
This section contains what is essentially a condensed first chapter of a book on finite groups \citep[see, e.g.,][]{rotman2012introduction}, and is provided for the benefit of readers who are not familiar with the topic. The terms and results covered here suffice to prove the main results of this paper.

Denote by $N = \{1,\ldots,n\}$. A {\em permutation} of $N$ is a bijection $g \colon N \to N$. The inverse of a permutation $g$ is denoted by $g^{-1}$ (so $g^{-1}(g(i)) = i$), and the composition of two permutations $g$ and $h$ is simply $g h$; i.e., if $k = g h$ then $k(i) = g(h(i))$.

A {\em group}---for our purposes---will be a non-empty set of permutations that (1) contains $g^{-1}$ whenever it contains $g$, and (2) contains $g h$ whenever it contains both $g$ and $h$. It follows from this definition that every group must include the trivial, identity permutation $e$ that satisfies $e(i)=i$ for all $i$.

Groups often appear as sets of permutations that {\em  preserve some invariant}. In our case, $\Aut_f$ is the group of permutations of the voters that preserves every outcome of $f$. It is easy to see that $\Aut_f$ is indeed a group.

A subgroup $H$ of $G$ is simply a subset of $G$ that is also a group. Given $g \in G$, we denote 
$$
  g H = \{g h \in G\, :\, h \in H\}.
$$  
The sets $g H$ are in general not subgroups, and are called the {\em left cosets} of $H$ (the right cosets are of the form $H g$). It is easy to verify that all left cosets are disjoint, and that each has the same size as $H$. It follows that the size of $G$ is divisible by the size of $H$.

Given an element $i \in N$, we denote by $G_i$ the set of permutations that fix $i$. That is, $g \in G_i$ if $g(i) =i$. $G_i$ is a subgroup of $G$. It is called the {\em stabilizer} of $i$.

The {\em $G$-orbit} of $i \in N$ is the set of $j \in N$ such that $j = g(i)$ for some $g \in G$, and is denoted by $G i$. As it turns out, if $j$ is in the orbit of $i$ then the set of $g \in G$ such that $g(i)=j$ is a coset of the stabilizer $G_i$. It follows that there is a bijection between the orbit $G i$ and the cosets of $G_i$. This is called the {\em Orbit-Stabilizer Theorem}.

Recall that $G$ is transitive if for all $i,j$ there is a $g \in G$ such that $g(i)=j$. This is equivalent to there existing only a single $G$-orbit, or that $j$ is in the orbit of $i$ for every $i,j$. Therefore, if $G$ is transitive, the orbit $G i$ is of size $n$, and since we can identify this orbit with the cosets of $G_i$, there are $n$ such cosets. Since they are all the same size as $G_i$, and since they form a partition of $G$, each coset of $G_i$ must be of size $|G|/n$. We will use this fact in the proof of Theorem~\ref{thm:transitive-lower}.

\section{Proofs}
\label{sec:proofs}
\subsection{Proofs of Propositions~\ref{clm:anonymous},~\ref{prop:identical},~\ref{clm:role-transitive} and~\ref{clm:role-k-transitive}}
\begin{proof}[Proof of Proposition \ref{clm:anonymous}]
We begin by showing that if $f$ is anonymous, then all assignments are equivalent under $f$.

Assume that $f$ is symmetric, and let $a$ and $b$ be two assignments. Define $\sigma = b \circ a^{-1}$. Then, by definition of symmetry, for every $\phi \in X^V$,
\begin{align*}
    f_a(\phi) = f(\phi \circ a^{-1})
    = f((\phi \circ a^{-1})^{\sigma})
    = f(\phi \circ a^{-1} \circ \sigma^{-1})
    = f(\phi \circ a^{-1} \circ (a \circ b^{-1}))
    = f(\phi \circ b^{-1})
    = f_b(\phi)
\end{align*}
and hence $a$ and $b$ are equivalent. Since $a$ and $b$ were arbitrary, it follows that all assignments are equivalent under $f$.


Conversely, suppose that all assignments are equivalent under $f$. We need to show that for any $\sigma \in S(R)$ and any $\phi \in X^R$, $f(\phi) = f(\phi^{\sigma})$. Fix $\sigma$ and $\phi$, let $a$ be any assignment, and let $b = \sigma \circ a$. Then since $a$ and $b$ are equivalent,
\begin{align*}
    f(\phi) = f((\phi \circ a) \circ a^{-1})
    = f_a(\phi \circ a)
    = f_b(\phi \circ a)
    = f(\phi \circ (a \circ b^{-1}))
    = f(\phi \circ \sigma^{-1})
    = f(\phi^{\sigma}).
\end{align*}
Since $\sigma$ and $\phi$ were arbitrary, it follows that $f$ is symmetric.
\end{proof}

In analogy with the notion of equivalence of roles given in Section~\ref{sec:roles}, say that two $k$-tuples of roles $\mathbf{r}$ and $\mathbf{s}$ are equivalent if there is a $k$-tuple $\mathbf{v}$ of voters and a pair of assignments $a$ and $b$ such that $f_a = f_b$, $a(\mathbf{v}) = \mathbf{r}$, and $b(\mathbf{v}) = \mathbf{s}$. Proposition~\ref{prop:identical} then follows from the next proposition by setting $k$ to $1$: 

\begin{customprop}{3}
\label{prop:k-identical}
An abstract voting rule $f \colon X^R \to X$ is $k$-equitable if and only if all $k$-tuples of roles are equivalent under $f$.
\end{customprop}

\begin{proof}
Fix an assignment $a$. The rule $f$ is $k$-equitable if and only if $f_a$ is $k$-equitable, which holds if and only if for every pair of $k$-tuples $(v_1, \dots, v_k)$ and $(w_1, \dots, w_k)$ of voters there is a $\sigma \in S(V)$ such that $\sigma((v_1, \dots, v_k)) = (w_1, \dots, w_k)$ and for all $\phi \in X^V$, $f_a(\phi) = f_a(\phi^{\sigma})$. Since
\begin{align*}
    f_a(\phi^{\sigma}) &= f((\phi \circ \sigma^{-1}) \circ a^{-1})\\
    &= f(\phi \circ (a \circ \sigma)^{-1})\\
    &= f_{a \circ \sigma}(\phi),
\end{align*}
it follows that $f$ is $k$-equitable if and only if for every pair of $k$-tuples $(v_1, \dots, v_k)$ and $(w_1, \dots, w_k)$ of voters there is a $\sigma \in S(V)$ such that $\sigma((v_1, \dots, v_k)) = (w_1, \dots, w_k)$ and $f_a = f_{a \circ \sigma}$.

Now, this holds if and only if for every pair of $k$-tuples $(r_1, \dots, r_k)$ and $(s_1, \dots, s_k)$ of roles there is a $\sigma \in S(V)$ such that $\sigma(a^{-1}(r_1, \dots, r_k)) = a^{-1}(s_1, \dots, s_k)$ and $f_a = f_{a \circ \sigma}$. But this holds if and only if for every such pair of $k$-tuples of roles, there is a $\sigma \in S(V)$ and a $k$-tuple of voters $(v_1, \dots, v_k)$ such that $f_a = f_{a \circ \sigma}$, $a((v_1, \dots, v_k)) = (r_1, \dots, r_k)$, $(a \circ \sigma)((v_1, \dots, v_k)) = (s_1, \dots, s_k)$, which holds if and only if every pair of $k$-tuples of roles is equivalent under $f$.
\end{proof}

We now prove Proposition~\ref{clm:role-k-transitive}. Proposition~\ref{clm:role-transitive} then follows directly by setting $k$ to $1$.

\begin{proof}[Proof of Proposition~\ref{clm:role-k-transitive}]
We first show that if $f$ is $k$-equitable, then there is a set $A$ as above. 

Assume $f$ is $k$-equitable. Fix an assignment $a$, and let
\begin{align*}
    A = \{ b \text{ an assignment s.t. } f_a = f_b \}. 
\end{align*}
(1) is immediate from the definition of $A$. To see that (2) holds, note that for any $k$-tuple of roles $(r_1, \dots, r_k)$ and any $k$-tuple of voters $(v_1, \dots, v_k)$, it follows from a result analogous to Proposition~\ref{prop:k-identical} that since $f$ is $k$-equitable, there are assignments $c$ and $d$ such that $c((v_1, \dots, v_k)) = (r_1, \dots, r_k)$, $d((v_1, \dots, v_k)) = a((v_1, \dots, v_k))$, and $f_c = f_d$. Now, $f_c = f_d$ implies that for all $\phi \in X^V$, $f_c(\phi) = f_d(\phi)$, which implies that for all $\phi \in X^V$, $f_c(\phi \circ (a^{-1} \circ d)) = f_d(\phi \circ (a^{-1} \circ d))$. It follows that, for all $\phi \in X^V$,
\begin{align*}
    f_{c \circ d^{-1} \circ a}(\phi) &= f(\phi \circ (c \circ d^{-1} \circ a)^{-1})\\
    &= f((\phi \circ (a^{-1} \circ d)) \circ c^{-1})\\
    &= f_c(\phi \circ (a^{-1} \circ d))\\
    &= f_d(\phi \circ (a^{-1} \circ d))\\
    &= f((\phi \circ (a^{-1} \circ d)) \circ d^{-1})\\
    &= f(\phi \circ a^{-1})\\
    &= f_a(\phi),
\end{align*}
and hence $f_{c \circ d^{-1} \circ a} = f_a$. But this implies that $c \circ d^{-1} \circ a \in A$. Since $(c \circ d^{-1} \circ a)((v_1, \dots, v_k)) = c(d^{-1}(a((v_1, \dots, v_k)))) = c((v_1, \dots, v_k)) = (r_1, \dots, r_k)$, (2) then follows.

We now show that if there is a set $A$ as above, then $f$ is $k$-equitable.

Assume there is such an $A$. By a result analogous to Proposition~\ref{prop:k-identical}, it is sufficient to show that all $k$-tuples of roles are equivalent under $f$. But this follows immediately from (2) and the definition of equivalence of $k$-tuples of roles.
\end{proof}

\subsection{Proof of Lemma~\ref{lem:make-voting-rule-from-winning-coalitions}}

\begin{proof}[Proof of Lemma~\ref{lem:make-voting-rule-from-winning-coalitions}]
Let $f$ be the voting rule defined as follows. For a voting profile $\phi$, if there is a set $W \in \mathcal{W}$ such that $\phi(w) = 1$ for all $w \in W$, then $f(\phi) = 1$, and similarly, if there is a set $W \in \mathcal{W}$ such that $\phi(w) = -1$ for all $w \in W$, then $f(\phi) = -1$. This is well-defined, since if there are two such sets $W$, they must agree because they intersect. If there are no such sets, then $f(\phi)$ is determined by majority.

That $f$ is neutral follows immediately from the symmetry in the definition of $f$ when some $W \in \mathcal{W}$ agrees on either $1$ or $-1$ and the fact that majority is neutral. To see that $f$ is positively responsive, suppose that $f(\phi) \in \{0, 1\}$, $\phi'(x) \geq \phi(x)$ for all $x \in V$, and $\phi'(y) > \phi(y)$ for some $y \in V$. Since $f(\phi) \neq -1$, there is no set $W \in \mathcal{W}$ such that $\phi(x) = -1$ for all $x \in W$, hence the same is true for $\phi'$. If there is some set $W \in \mathcal{W}$ such that $\phi'(x) = 1$ for all $x \in W$, then $f(\phi') = 1$. If not, then the same is true of $\phi$, and hence by positive responsiveness of majority, $f(\phi') = 1$.

Finally, it is immediate from the definition of $f$ that every $W \in \mathcal{W}$ is a winning coalition.
\end{proof}

\subsection{Proof of Theorem~\ref{thm:transitive-upper}}
\label{sec:longest-run}

\newcommand{\roof}[1]{\lceil #1 \rceil}

\begin{proof}[Proof of Theorem~\ref{thm:transitive-upper}]
The longest-run voting rule is equitable, since any rotation of the cycle is an automorphism. That is, for every $k$, the map $\sigma \colon V \to V$ defined by $\sigma(i) = i + k \text{ mod } n$ leaves the outcome unchanged. Furthermore, for every pair of voters $i,j$, if we set $k = i-j$, then $\sigma(i)=j$.

For every $V$ of size $n$, we claim that the longest-run rule $\ell \colon X^V \to X$ has winning coalitions of size $2\roof{\sqrt{n}\;}-1$.

Let 
$$W = \{0, \dots, \roof{\sqrt{n}\;}-1\} \cup \{w\,:\, w \text{ mod } \roof{\sqrt{n}\;} = 0\}.$$

Any run that is disjoint from $W$ is of length at most $\roof{\sqrt{n}\;}-1$ since the second set in the definition of $W$ is comprised of voters who are at most $\roof{\sqrt{n}\;}$ apart. However, the first set is a contiguous block of length $\roof{\sqrt{n}\;}$. Hence, if all $w \in W$ vote identically in $\{-1,1\}$, the longest run will be a subset of $w$, and hence the outcome will be the vote cast by all members of $W$. It then follows that $W$ is a winning coalition.

Finally,
\begin{align*}
    |W| &= |\{0, \dots, \roof{\sqrt{n}\;}-1\}| + |\{w\,:\, w \text{ mod } \roof{\sqrt{n}\;} = 0\}| - |\{0\}|\\
    &= \roof{\sqrt{n}\;} + |\{w\,:\, w \text{ mod } \roof{\sqrt{n}\;} = 0\}| - 1\\
    &\leq 2 \roof{\sqrt{n}\;} - 1.
\end{align*}
Since every superset of a winning coalition is again a winning coalition, the result follows.
\end{proof}

\subsection{Proof of Theorem~\ref{thm:transitive-lower}}
Readers who are not familiar with the theory of finite groups are encouraged to read \S\ref{sec:primer} before reading this proof.

Recall that the group of all permutations of a set of size $n$ is denoted by $S_n$.

The next lemma shows that if a group $G$ acts transitively on $\{1,\ldots,n\}$, then any set $S$ that intersects all of its translates (i.e., sets of the form $g S$ for $g \in G$) must be of size at least $\sqrt{n}$. The proof of the theorem will apply this lemma to a winning coalition $S$.
\begin{lemma}
\label{prop:lower-bound-in-group-action-setting}
Let $G \subset S_n$ be transitive, and suppose that $S \subseteq V$ is such that for all $g \in G$, $gS \cap S \neq \emptyset$. Then $|S| \geq \sqrt{n}$.
\end{lemma}
\begin{proof}
For any $v, w \in V$, define $\Gamma_{v, w} = \{g \in G : g(v) = w\}$. 
Then $\Gamma_{v, w}$ is a left coset of the stabilizer of $v$. Hence, and since the action is transitive, it follows from the Orbit-Stabilizer Theorem that $|\Gamma_{v, w}| = \frac{|G|}{n}$. If $gS \cap S \neq \emptyset$ for all $g \in G$, then for any $g \in G$, there exists $v, w \in S$ such that $g(v) = w$, hence 
\begin{align*}
	\bigcup_{v, w \in S} \Gamma_{v, w} = G. 
\end{align*}
So
\begin{align*}
	|G| &= \left|\bigcup _{v, w} \Gamma_{v, w}\right| \le \sum _{v, w} |\Gamma_{v, w}| = |S|^2\frac{|G|}{n},
\end{align*}
and we conclude that $|S| \ge \sqrt{n}$.  
\end{proof}
	
Our lower bound (Theorem~\ref{thm:transitive-lower}) is an immediate corollary of this claim. 

\begin{proof}[Proof of Theorem~\ref{thm:transitive-lower}]
Let $f$ be an equitable voting rule for the voter set $V$. Suppose that $W \subseteq V$ is a winning coalition for $f$. Then, for every $\sigma \in \Aut_f$, it must be the case that $\sigma(W) \cap W \neq \emptyset$ (otherwise, $f$ would not be well-defined). Hence, it follows from Lemma~\ref{prop:lower-bound-in-group-action-setting} that $|W| \geq \sqrt{n}$.
\end{proof}

\subsection{Proof of Theorem~\ref{thm:college}}

\begin{proof}[Proof of Theorem~\ref{thm:college}]
Define $C(n)$ to be the smallest size of any winning coalition in any generalized representative democracy rule for $n$ voters. We want to show that $C(n) \ge n^{\log_3 2}$. 

If $n = 1$, a winning coalition must be of size $1$, which is $\ge 1^{\log_3 2}$. 

If $n > 1$, any generalized voting rule $f$ is of the form $f(\phi) = \maj(f_1(\phi |_{V_1}), f_2(\phi |_{V_2}),\dots,f_d(\phi |_{V_d}))$.
Because the voting rule is equitable, the functions $f_i$ are all isomorphic, and so have minimal winning coalitions of the same size. A minimal winning coalition for $f$ would then need to include a strict majority of these, which is of size at least $\frac{d+1}{2}$. Therefore,\footnote{Here and below $d|n$ denotes that $d$ is a divisor of $n$.} 
\begin{align}
  \label{eq:Cn}
  C(n) \geq \min_{d|n} \frac{d+1}{2} \cdot C\left(n/d\right).
\end{align}

Assume by induction that $C(m) \ge m^{\log_3 2}$ for all $m < n$. Then for $d|n$, 
\begin{align*}
\frac{d+1}{2}\cdot C\left(n/d\right) \ge \frac{d+1}{2} \cdot \left(\frac{n}{d}\right)^{\log_3 2} = n^{\log_3 2}\cdot\frac{d+1}{2} \cdot d^{-\log_3 2}.
\end{align*}
Denote $h(d)= \frac{d+1}{2} \cdot d^{-\log_3 2}$, so that
$$
\frac{d+1}{2}C(n/d)\geq n^{\log_3 2}h(d).
$$
Note that $h(d) \geq 1$. To see this, observe that  $h(3) = 1$, and 
$$h'(d) = \frac{d^{-\log 6 / \log 3}(d \log \frac{3}{2} - \log 2)}{2 \log 3} > 0$$ 
for $d \ge 3$, and so $h(d) \ge 1$ for $d \ge 3$.

We have thus shown that
$$
\frac{d+1}{2}C(n/d)\geq n^{\log_3 2},
$$
and so by \eqref{eq:Cn}, $C(n) \geq n^{\log_3 2}$.

To see that $C(n) = n^{\log_3 2}$ for arbitrarily large $n$, consider the following GRD rule (see Figure~\ref{fig:college}).
Take $n$ to be a power of $3$, and let $f$ be defined recursively by partitioning at each level into three sets $\{V_1,V_2,V_3\}$ of equal size. A simple calculation shows that the winning coalition recursively consisting of the winning coalitions of any two of $\{V_1,V_2,V_3\}$ (e.g., $V_1$ and $V_3$, as in Figure~\ref{fig:college}), is of size $n^{\log_3 2}$.
\end{proof}
The construction of small winning coalitions in the last part of the proof mimics the construction of the Cantor set.

\subsection{Proof of Theorems \ref{thm:2-eq-neg} and \ref{thm:6-transitive}}
\label{sec:proofs-thm26}
The group of all even permutations is called the {\em alternating group} and is denoted $A_n$.
\begin{lemma}
\label{lem:alternating-or-symmetric-then-big-winning-coalitions}
Let $f$ be a voting rule for $n$ voters. If $\Aut_f$ is either $S_n$ or $A_n$ then every winning coalition for $f$ has size at least $n/2$.
\end{lemma}

\begin{proof}
Suppose $W \subseteq V$ is a winning coalition for $f$ with $|W| = k < n/2$. Label the voters $V$ with labels $1, \dots, n$ such that $W = \{1, \dots, k\}$, and let $\pi$ be the permutation of $V$ given by $\pi(i) = n + 1 - i$ for $i = 1, \dots, n$. If $\lfloor n/2 \rfloor$ is odd, let $\pi$ be the map above composed with the map that exchanges $1$ and $2$. It follows that $\pi$ is in the alternating group, and hence $\pi \in \Aut_f$. However, $\pi(W) \cap W = \emptyset$ since $k < n + 1 - k$, contradicting the assumption $W$ is a winning coalition. 
\end{proof}

\begin{proof}[Proof of Theorem~\ref{thm:2-eq-neg}]
Denote by $\eta(n)$ the number of positive integers $m \leq n$ for which there is no $2$-transitive group action on a set of $m$ elements except for $S_m$ and $A_m$. It follows from the main theorem in \cite{cameron1982} that $n - \eta(n)$ is at most $3n/\log(n)$ for all $n$ large enough. Since
\begin{align*}
    \lim_{n \rightarrow \infty}{\frac{3n/\log(n)}{n}} = 0,
\end{align*}
it follows that
\begin{align*}
    \lim_{n \rightarrow \infty}{\frac{\eta(n)}{n}} = \lim_{n \rightarrow \infty}{1 - \frac{n - \eta(n)}{n}} = 1,
\end{align*}
and so the claim follows from Lemma \ref{lem:alternating-or-symmetric-then-big-winning-coalitions}.
\end{proof}

\begin{proof}[Proof of Theorem~\ref{thm:6-transitive}]

The only $4$- or $5$-transitive finite groups aside from the alternating and symmetric groups are the Mathieu groups, with the largest action on a set of size $24$ \citep{dixon1996permutation}. Hence, for $n>24$, every $4$- or $5$-transitive voting rule must have either $S_n$ or $A_n$ as an automorphism group. Furthermore, the {\em only} $6$-transitive groups are $S_n$ or $A_n$ \citep[again, see][]{dixon1996permutation}. Hence, the result follows immediately from Lemma~\ref{lem:alternating-or-symmetric-then-big-winning-coalitions}. 
\end{proof}

\subsection{Proof of Proposition~\ref{prop:pd}}
\label{app:pd}
This proposition's proof follows directly from group theory results that are classical, but that are not covered in our primer in \S\ref{sec:primer}, and which we now review briefly. 

Let $G$ be a group. The {\em order} of $G$ is simply its size. The order of $g \in G$ is the smallest $n$ such that $g^n$ is the identity. Given a prime $p$, we say that a group $P$ is a $p$-group if the orders of all of its elements are equal to powers of $p$. We assume for the remainder of this section that $p$ is prime.

Let $G$ be a finite group with $|G|=p^d \cdot m$, where $d \geq 1$, and $m$ is not divisible by $p$. \cite{sylow1872theoremes} proved that, in this case, $G$ has a subgroup $P$ that is a $p$-group of order $p^d$. Such groups are called Sylow $p$-groups in his honor. The following lemma states an important and well-known fact \citep[see, e.g.,][Theorem~3.4']{wielandt2014finite} regarding Sylow $p$-groups.
\begin{lemma}
  \label{lem:sylow-transitive}
  Let $G$ act transitively on a set $V$ of size $n=p^d$ for some $d\geq 1$. Any Sylow $p$-subgroup of $G$ acts transitively on $V$.
\end{lemma}
\begin{proof}
  Since the order of $G$ is divisible by the size of $V$, $|G| = p^{d+\ell} \cdot m$ for some $\ell \geq 0$ and $m$ not divisible by $p$. Let $P$ be a Sylow $p$-subgroup, so that $|P| =  p^{d+\ell}$.

  For any $i \in V$, the size of the $P$-orbit $P i$ divides $|P|=p^{d+\ell}$, and so is equal to $p^{d-a}$ for some $a \geq 0$. Now, the size of the stabilizers $P_i$ and $G_i$ is $|P_i| = |P|/|Pi| = p^{\ell+a}$ and $|G_i| = p^\ell \cdot m$.  Since $P_i$ is a subgroup of $G_i$,  $|P_i|$ divides $|G_i|$, and so $a=0$, $|P i| = p^d = n$, and $P$ acts
  transitively on $V$.
\end{proof}
The {\em center} of a group is the collection of all of its elements that commute with all the group elements: $\{g \in G\,:\,g h=h g \text{ for all } h \in G\}$. This is easily seen to also be a subgroup of $G$. 
\begin{lemma}
\label{lem:p-group-center}
Every non-trivial $p$-group has a non-trivial center.
\end{lemma}
For a proof, see Theorem 6.5 in \cite{Lang2002}. Here and below, a non-trivial group is a group of order larger than 1. 

Let $G$ be a group that acts transitively on a set $V$, and let $Z$ be the center of $G$. Denote by $G/Z$ the set of left cosets of $Z$, and let $\hat V$ be the set of $Z$-orbits of $V$. If $v, w \in V$ are in the same $Z$-orbit, then $g(v)$ and $g(w)$ are also in the same $Z$-orbit, since if $z(v) = w$ then $z(g(v)) = g(z(v)) = g(w)$. Hence, each $g \in G$ induces a permutation on $\hat V$. Note that $g,h \in G$ induce the same permutation on $\hat V$ if they are in the same element of $G/Z$. Hence, we can think of $G/Z$ as a group of permutations of $\hat V$. This group must act transitively on $\hat V$ since $G$ acts transitively on $V$. Furthermore, if a subgroup of $G/Z$ acts transitively on $\hat V$, then the union of the cosets it includes is a subgroup of $G$ that acts transitively on $V$.
 

\begin{lemma}
\label{lem:transitive-p-group}
Let $G$ be a group acting transitively on a set $V$ of size $p^d$, for some $d\geq 1$. There exists a $p$-subgroup $R$ of $G$ of size $p^d$ acting transitively on $V$ with trivial stabilizers.
\end{lemma}
\begin{proof}
Let $P$ denote a Sylow $p$-subgroup of $G$. By Lemma~\ref{lem:sylow-transitive}, the action of $P$ on $V$ is also transitive.

Let $Z$ denote the center of $P$. Since $P$ is non-trivial, by Lemma~\ref{lem:p-group-center}, $Z$ is non-trivial. We claim that the $Z$ action on $V$ has trivial stabilizers. To see this, assume that $h(v) = v$ for some $v$, and choose any $w \in V$. Since $P$ acts transitively, there is some $g \in P$ such that $g(v)=w$. Since $h$ commutes with $g$, 
$$
  h(w) = h(g(v)) = g(h(v)) = g(v) = w,
$$
and so, since $w$ was arbitrary, $h$ is the identity. Hence, $Z_v = \{e\}$ for every $v \in V$. Note that by the Orbit-Stabilizer Theorem, this implies that each $Z$ orbit is equal in size to $Z$.

If the action of $Z$ is also transitive, we are done, since we can take $R=Z$.

Finally, consider the case that $Z$ does not act transitively.

In this case $P' = P/Z$ acts transitively on $\hat V$, the set of the $Z$-orbits of $V$. By induction, $P'$ has a subgroup $Z'$ which acts transitively with trivial stabilizers on $\hat V$, and hence has size $|\hat V|$. Note that since the action of $Z$ has trivial stabilizers, every $Z$-orbit has size $|Z|$, and so $|\hat V|= |V| / |Z| = p^d / |Z|$. Thus, taking $R$ to be the union of the cosets in $P/Z$ that comprise $Z'$, $R$ acts transitively on $V$. Finally, since this subgroup has size $|Z'| \cdot |Z| = p^d = |V|$, it follows that the action of $R$ has trivial stabilizers.




\end{proof}

A group is said to be {\em abelian}  if all of its elements commute: $g h = h g$ for all $g,h \in G$. Note that the center of every group is abelian by definition. The structure of abelian groups is simple and well understood: Kronecker's Theorem \citep[][Theorem 5.2.2]{kronecker1870auseinandersetzung,stillwell2012classical} states that every abelian group is a product of cycles of prime powers. That is, if $G$ is an abelian group of permutations of a set $V$---and assuming without loss of generality that no element of $V$ is fixed by all elements of $G$---then there is a way to identify $V$ with $\prod_{i=1}^m\{0,\ldots,n_i-1\}$, with each $n_i$ a prime power, so that $G$ is generated by permutations\footnote{A group is said to be generated by a set $S$ of permutations if it includes precisely those permutations that can be constructed by composing permutations in $S$.} of the form 
$$
\sigma_k(i_1,\ldots,i_m) = (i_1,\ldots,i_{k-1},i_k+1 \text{ mod } n_k,i_{k+1},\ldots,i_m).
$$

As is also well known, every group of order $p$ or $p^2$ is abelian \cite[page 148]{netto1892theory}. From these facts follows the next lemma.
\begin{lemma}
\label{lem:abelian}
Let $R$ be a group of order $p$ or $p^2$, acting transitively on a set $V$. In the former case, we can identify $V$ with $\{0,\ldots,p-1\}$ so that $R$ is generated by 
$$
  \sigma(i) = i+1 \text{ mod } p.
$$ 
In the latter case, we can either identify $V$ with $\{0,\ldots,p^2-1\}$ so that  $R$ is generated by 
$$
  \sigma(i) = i+1\text{ mod } p^2,
$$
or else we can identify $V$ with $\{0,\ldots,p-1\}^2$, so that $R$ is generated by 
\begin{align*}
  \sigma_1(i_1,i_2) = (i_1+1 \text{ mod } p,i_2)\\
  \sigma_1(i_1,i_2) = (i_1,i_2+1\text{ mod } p).
\end{align*}
\end{lemma}
Hence, if $R$ is a subgroup of the automorphism group of a voting rule $f$, then this rule is  cyclic if $n=p$, and is either cyclic, $2$-cyclic or both if $n=p^2$.

\begin{proof}[Proof of Proposition~\ref{prop:pd}]
By Lemma~\ref{lem:transitive-p-group}, $\Aut_f$ has a $p$-subgroup $R$ of  order $n$ that acts transitively on $V$. The claim now follows immediately from Lemma~\ref{lem:abelian}.
\end{proof}
When $n=p^d$, with $d>2$, this proof fails since the group $R$ is no longer necessarily abelian. Non-abelian groups do not have cyclic structure, and thus voting rules for which this group $R$ is not abelian will not be cyclic, $2$-cyclic, or higher-dimensional cyclic. We conjecture that such equitable voting rules do indeed exist.



\section{$\mathbf{2}$-Equitable and $\mathbf{3}$-Equitable Rules}
\label{sec:2-3-equitable}
In this section we construct $2$-equitable and $3$-equitable rules with small winning coalitions that apply to arbitrarily large population sizes. This construction is rather technically involved and uses finite vector spaces. To glean some intuition, we first explain an analogous construction using standard vector spaces and assuming a continuum of voters. 


Suppose voters are identified with the set of one-dimensional subspaces of $\mathbb{R}^3$: i.e., each voter is identified with a line that passes through the origin. Now suppose winning coalitions are the two-dimensional subspaces: if all voters on a plane agree, that is the election outcome, otherwise the election is undecided.\footnote{This rule is well defined since every pair of two-dimensional subspaces intersects, and so no two winning coalitions are disjoint.} Clearly, the winning coalitions are much smaller than the electorate (or indeed of ``half of the voters'') in the sense that they have a smaller~dimension.

Invertible linear transformations of $\mathbb{R}^3$ permute the one-dimensional subspaces, and the two-dimensional subspaces, and so constitute automorphisms of this voting rule. Equity follows since for any two non-zero vectors $v$ and $u$, we can find some invertible linear transformation that maps $v$ to $u$. Moreover, the voting rule is also $2$-equitable---given a pair of distinct voters $(v_1,v_2)$, and given another such pair $(u_1,u_2)$, we can find some invertible linear transformation that maps the former to the latter. Thus, \textit{every pair} of voters plays the same role.

In Theorem~\ref{thm:2-transitive} below we construct $2$-equitable voting rules for finite sets of voters, using finite vector spaces instead of $\mathbb{R}^3$. Figure~\ref{fig:pp} shows a $2$-equitable voting rule constructed in this way, for $7$ voters. In the figure, every three co-linear nodes form a winning coalition, as well as the three nodes on the circle.\footnote{Figure~\ref{fig:pp} depicts what is commonly referred to as a Fano plane in finite geometry. It is the finite projective plane of order 2.} In this construction, the size of the winning coalition is exactly $\sqrt{n}$ (rounded up to the nearest integer), which matches the lower bound of $\sqrt{n}$ in  Theorem~\ref{thm:transitive-lower}.

\begin{figure}
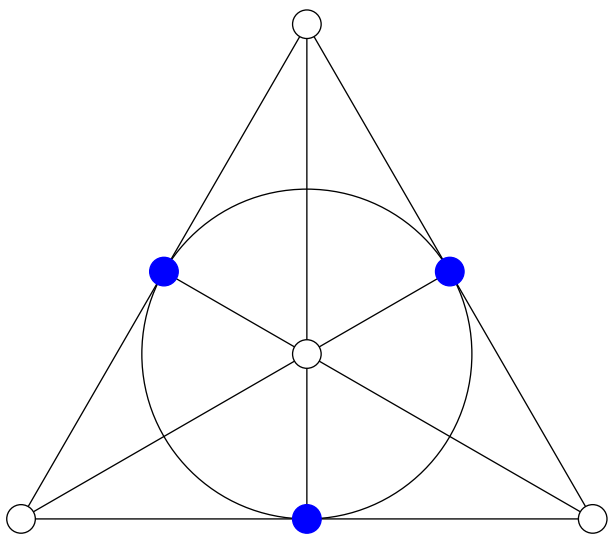

\begin{center}
\begin{asy}
size(8cm,0);
real circlesize = 0.05;  

pair p001 = (0,0);
pair p010 = (2,0);
pair p100 = (1,sqrt(3));
pair p011 = (1,0);
pair p101 = (1/2,sqrt(3)/2);
pair p110 = (3/2,sqrt(3)/2);
pair p111 = (1,sqrt(3)/3);

draw(p001 -- p010 -- p100 -- p001);
draw(p001 -- p110);
draw(p010 -- p101);
draw(p100 -- p011);
draw(circle(p111,sqrt(3)/3));

fill(circle(p001,circlesize), white); draw(circle(p001,circlesize));
fill(circle(p010,circlesize), white); draw(circle(p010,circlesize));
fill(circle(p100,circlesize), white); draw(circle(p100,circlesize));
fill(circle(p011,circlesize), blue); draw(circle(p011,circlesize), blue);
fill(circle(p101,circlesize), blue); draw(circle(p101,circlesize), blue);
fill(circle(p110,circlesize), blue); draw(circle(p110,circlesize), blue);
fill(circle(p111,circlesize), white); draw(circle(p111,circlesize));
\end{asy}
\end{center}
\caption{Every three co-linear points form a winning coalition, as well as the three points on the circle (marked in blue). This voting rule is $2$-equitable.}
\label{fig:pp}
\end{figure}

\begin{theorem}
\label{thm:2-transitive}
Let the set of voters be of size  $n = q^{2}+q+1$, for prime $q$. Then  there is a $2$-equitable voting rule with a winning coalition of size exactly equal to $\sqrt{n}$, rounded up to the nearest integer.
\end{theorem}
More generally, a similar statement holds when $n = q^{2}+q+1$ and $q=p^k$ for some $k\geq 1$ and $p$ prime. The example in Figure~\ref{fig:pp} corresponds to the case $q=2$.
\begin{proof}
  Let $\mathbb F_q$ denote the finite field with $q$ elements.\footnote{$\mathbb F_q$ is the set $\{0,1,\ldots,q-1\}$, equipped with the operations of addition and multiplication modulo $q$. The primality of $q$ is required to make multiplication invertible.} 
  
  Given a positive integer $m$, $\mathbb F_q^m$ is a vector space, where the scalars take values in $\mathbb F_q$: it satisfies all the axioms that (say) $\mathbb R^3$ satisfies, but for scalars that are in $\mathbb F_q$ instead of $\mathbb R$. Indeed, much of the standard theory of linear algebra of $\mathbb R^m$ applies in this finite setting, and we will make use of it here. 
  
  In particular, we will make use of $GL(m,q)$, the group of invertible, $m\times m$ matrices with entries in $\mathbb F_q$. Here, again, the product of two matrices is calculated as usual, but addition and multiplication are taken modulo $q$. Since $\mathbb F_q^m$ is finite, each matrix in $GL(m,q)$ corresponds to a permutation of $\mathbb F_q^m$. As in the case of matrix multiplication on $\mathbb R^m$, these permutations preserve the $1$-dimensional and $2$-dimensional subspaces. Moreover, this group acts $2$-transitively on the $1$-dimensional subspaces, as any two non-colinear vectors $(u,v)$ can be completed to a basis of $\mathbb F_q^m$, and likewise starting from $(u',v')$; then any basis can be carried by an invertible matrix to any other basis.
  
  With this established, we are ready to identify our set of voters with the set of $1$-dimensional subspaces of $\mathbb F_q^3$. For each $2$-dimensional subspace $U$ of $\mathbb F_q^3$, define the set $S_U$ of $1$-dimensional subspaces (i.e., voters) contained in $U$. Let $\mathcal{W}$ be the collection of all such sets $S_U$, and define, using Lemma~\ref{lem:make-voting-rule-from-winning-coalitions}, a  voting rule $f$ in which the sets $S_U$ are winning coalitions.
  We need to verify that any two winning coalitions $S_U$ and $S_{U'}$ are non-disjoint. This simply follows from the fact that every pair of $2$-dimensional subspaces intersects in some $1$-dimensional subspace, and so it follows that each pair of such winning coalitions will have exactly one voter in common.\footnote{This is the reason that the winning coalitions of this rule are so small and proving a tight match to the lower bound.}
  
  A simple calculation shows that the winning coalitions are of size $q+1$. Since $\sqrt{n} \leq q + 1 \leq \sqrt{n}+1$, the claim follows.
\end{proof}

To construct $3$-equitable rules we will need the following lemma. It allows us to show, using the probabilistic method, that for small automorphism groups we can construct voting rules with small winning coalitions. This is useful for proving that there exist $3$-equitable voting rules with small winning coalitions.
\begin{lemma}
\label{thm:probabilistic-approach}
Let $G$ be a group of $m$ permutations of $\{1,\ldots,n\}$. Then there is a neutral and positively responsive voting rule $f$ such that $G$ is a subgroup of $\Aut_f$, and $f$ has winning coalitions of size at most $2 \sqrt{n} \log m + 2$.
\end{lemma}
We use this lemma to prove our theorem illustrating the existence of $3$-equitable rules with small winning coalitions for arbitrarily large voter populations. We then return to prove the lemma.
\begin{theorem}
\label{thm:3-transitive}
For $n$ such that $n-1$ is a prime power, there is a $3$-equitable voting rule with a winning coalition of size at most $6 \sqrt{n} \log n$.
\end{theorem}

\begin{proof}
 For $n$ such that $n-1$ is the power of some prime there is a $3$-transitive group of permutations of $\{1,\ldots,n\}$ that is of size $m<n^3$.\footnote{The group $PGL(2, n-1)$ acts $3$-transitively on the projective line over the field $\mathbb F_{n-1}$, and is of size $n(n-1)(n-2) < n^3$.} Hence, by Lemma~\ref{thm:probabilistic-approach}, there is a $3$-equitable voting rule for $n$ (i.e., a rule with a $3$-transitive automorphism group) with a winning coalition of size at most $2 \sqrt{n} \log (n^3) = 6 \sqrt{n} \log n$.
\end{proof}
It is natural to conjecture that this probabilistic construction is not optimal, and that there exist $3$-equitable rules with winning coalitions of size $O(\sqrt{n})$.

The heart of  Lemma~\ref{thm:probabilistic-approach} is the following group-theoretic claim, which states that when $G$ is small then we can find a small set $S$ such that $g S$ and $S$ are non-disjoint for every $g \in G$. These sets $gS$ will be the winning coalitions used to prove Lemma~\ref{thm:probabilistic-approach}. The proof of this proposition uses the {\em probabilistic method}: we choose $S$ at random from some distribution, and show that, with positive probability, it has the desired property. This proves that there exists a deterministic $S$ with the desired property.
\begin{customprop}{4}
\label{prop:probabilistic}
Let a group $G$ of $m > 2$ permutations of $\{1,\ldots,n\}$. Then there exists a set $S \subseteq \{1,\ldots,n\}$ with $|S| \leq 2 \sqrt{n} \log m + 2$ such that $\forall g \in G$ we have $ gS \cap S \ne \emptyset$.  
\end{customprop}
\begin{proof}
To prove this, we will choose $S$ at random, and prove that it has the desired properties with positive probability. Let $\ell = \lceil \sqrt{n} \log |G| \rceil$. Let $S = S_1 \cup S_2$, where $S_1$ is any subset of $X$ of size $\ell$, and $S_2$ is the union of $\ell$ elements of $X$, chosen independently from the uniform distribution. Hence $S$ includes at most $2\ell \leq 2 \sqrt{n} \log |G| + 2$ elements.

We now show that $\mathbb{P}( \forall g \in G : gS \cap S \ne \emptyset) > 0$, and hence there is some set $S$ with the desired property. Note that for any particular $g \in G$, the distribution of $g S_2$ is identical to the distribution of $S_2$. Hence
\begin{align*}
    \mathbb{P}(gS \cap S = \emptyset) 
    &\leq \mathbb{P}(gS_2 \cap S_1 = \emptyset)\\
    &= \mathbb{P}(S_2 \cap S_1 = \emptyset)\\
    &= \left(\frac{n-\ell}{n}\right)^\ell\\ 
    &\leq e^{-\ell^2/n}\\
    &\leq e^{-(\log m)^2}.
\end{align*}
Thus, the probability that there is some $g \in G$ for which $gS \cap S = \emptyset$ is, by taking a union bound, at most
\begin{align*}
    m e^{-(\log |G|)^2},
\end{align*}
which is strictly less than $1$ for $m > 2$. 
\end{proof}
We are finally ready to prove Lemma~\ref{thm:probabilistic-approach}.
\begin{proof}[Proof of Lemma \ref{thm:probabilistic-approach}]
  Let $S$ be the subset of $\{1,\ldots,n\}$ given by Proposition~\ref{prop:probabilistic}. Let $\mathcal{W}$ be the collection of sets of the form $gS$, where $g \in G$. This is a collection of pairwise non-disjoint sets, since if $g S$ and $h S$ intersect then so do $h^{-1}g S$ and $S$, which is impossible by the defining property of $S$. Since $|S|=2 \sqrt{n} \log m + 2$ the claim follows from Lemma~\ref{lem:make-voting-rule-from-winning-coalitions}.
\end{proof}

\newpage




\bibliography{refs}

\begin{thebibliography}{27}
\providecommand{\natexlab}[1]{#1}
\providecommand{\url}[1]{\texttt{#1}}
\expandafter\ifx\csname urlstyle\endcsname\relax
  \providecommand{\doi}[1]{doi: #1}\else
  \providecommand{\doi}{doi: \begingroup \urlstyle{rm}\Url}\fi

\bibitem[Aschbacher et~al.(2011)Aschbacher, Lyons, Smith, and
  Solomon]{aschbacher2011classification}
M.~Aschbacher, R.~Lyons, S.~D. Smith, and R.~Solomon.
\newblock \emph{The Classification of finite simple groups: Groups of
  characteristic 2 type}.
\newblock Number 172. American Mathematical Soc., 2011.

\bibitem[Bhatnagar(2020)]{bhatnagar2020voting}
A.~Bhatnagar.
\newblock {Voting Rules that are Unbiased but not Transitive-Symmetric}.
\newblock \emph{{The Electronic Journal of Combinatorics}}, 27, 2020.

\bibitem[Cameron et~al.(1982)Cameron, Neumann, and Teague]{cameron1982}
P.~J. Cameron, P.~M. Neumann, and D.~N. Teague.
\newblock On the degrees of primitive permutation groups.
\newblock \emph{Mathematische Zeitschrift}, 180\penalty0 (3):\penalty0
  141--149, 1982.
\newblock ISSN 1432-1823.

\bibitem[Cantillon and Rangel(2002)]{cantillon2002}
E.~Cantillon and A.~Rangel.
\newblock A graphical analysis of some basic results in social choice.
\newblock \emph{Social Choice and Welfare}, 19\penalty0 (3):\penalty0 587--611,
  2002.

\bibitem[Dasgupta and Maskin(1986)]{dasgupta1986existence}
P.~Dasgupta and E.~Maskin.
\newblock The existence of equilibrium in discontinuous economic games, i:
  Theory.
\newblock \emph{Review of Economic Studies}, 53\penalty0 (1):\penalty0 1--26,
  1986.

\bibitem[Dixon and Mortimer(1996)]{dixon1996permutation}
J.~D. Dixon and B.~Mortimer.
\newblock \emph{Permutation Groups}, volume 163.
\newblock Springer Science \& Business Media, 1996.

\bibitem[Dubey and Shapley(1979)]{dubey1979mathematical}
P.~Dubey and L.~S. Shapley.
\newblock Mathematical properties of the banzhaf power index.
\newblock \emph{Mathematics of Operations Research}, 4\penalty0 (2):\penalty0
  99--131, 1979.

\bibitem[Fey(2004)]{fey2004}
M.~Fey.
\newblock May's theorem with an infinite population.
\newblock \emph{Social Choice and Welfare}, 23\penalty0 (2):\penalty0 275--293,
  2004.

\bibitem[Goodin and List(2006)]{goodwin2006}
R.~Goodin and C.~List.
\newblock A conditional defense of plurality rule: Generalizing may's theorem
  in a restricted informational environment.
\newblock \emph{American Journal of Political Science}, 50\penalty0
  (4):\penalty0 940--949, 2006.

\bibitem[Isbell(1960)]{isbell1960homogeneous}
J.~R. Isbell.
\newblock Homogeneous games. {I}{I}.
\newblock \emph{Proceedings of the American Mathematical Society}, 11\penalty0
  (2):\penalty0 159--161, 1960.

\bibitem[Kronecker(1870)]{kronecker1870auseinandersetzung}
L.~Kronecker.
\newblock \emph{Auseinandersetzung einiger {E}igenschaften der {K}lassenzahl
  idealer complexer {Z}ahlen}.
\newblock 1870.

\bibitem[Lang(2002)]{Lang2002}
S.~Lang.
\newblock \emph{Algebra}.
\newblock Springer New York, 2002.

\bibitem[May(1952)]{may1952set}
K.~O. May.
\newblock A set of independent necessary and sufficient conditions for simple
  majority decision.
\newblock \emph{Econometrica}, 20\penalty0 (4):\penalty0 680--684, 1952.

\bibitem[McGann et~al.(2016)McGann, Smith, Latner, and
  Keena]{gerrymandering2016}
A.~J. McGann, C.~A. Smith, M.~Latner, and A.~Keena.
\newblock \emph{Gerrymandering in America: The House of Representatives, the
  Supreme Court, and the Future of Popular Sovereignty}.
\newblock Cambridge University Press, 2016.

\bibitem[Mossel and O'Donnell(1998)]{mossel1998recursive}
E.~Mossel and R.~O'Donnell.
\newblock Recursive reconstruction on periodic trees.
\newblock \emph{Random Structures \& Algorithms}, 13\penalty0 (1):\penalty0
  81--97, 1998.

\bibitem[Nash(1951)]{nash1951non}
J.~Nash.
\newblock Non-cooperative games.
\newblock \emph{Annals of Mathematics}, pages 286--295, 1951.

\bibitem[Netto(1892)]{netto1892theory}
E.~Netto.
\newblock \emph{The theory of substitutions and its applications to algebra}.
\newblock Inland Press, 1892.
\newblock Available online at
  \url{https://archive.org/details/theoryofsubstitu00nett/page/148}.

\bibitem[Packel(1980)]{packel1980transitive}
E.~W. Packel.
\newblock Transitive permutation groups and equipotent voting rules.
\newblock \emph{Mathematical Social Sciences}, 1\penalty0 (1):\penalty0
  93--100, 1980.

\bibitem[Peitgen et~al.(1993)Peitgen, J\"urgens, and
  Saupe]{peitgen1993fractals}
H.-O. Peitgen, H.~J\"urgens, and D.~Saupe.
\newblock \emph{Chaos and Fractals: New Frontiers of Science}.
\newblock Springer, 1993.

\bibitem[Penrose(1946)]{penrose1946}
L.~Penrose.
\newblock The elementary statistics of majority voting.
\newblock \emph{Journal of the Royal Statistical Society}, 109:\penalty0
  53--57, 1946.

\bibitem[Reiker(1962)]{reiker1962}
W.~Reiker.
\newblock \emph{The Theory of Political Coalitions}.
\newblock New Haven and London: Yale University Press, 1962.

\bibitem[Rotman(2012)]{rotman2012introduction}
J.~J. Rotman.
\newblock \emph{An Introduction to the Theory of Groups}, volume 148.
\newblock Springer Science \& Business Media, 2012.

\bibitem[Stillwell(2012)]{stillwell2012classical}
J.~Stillwell.
\newblock \emph{Classical topology and combinatorial group theory}, volume~72.
\newblock Springer Science \& Business Media, 2012.

\bibitem[Sylow(1872)]{sylow1872theoremes}
M.~Sylow.
\newblock Th\'eor\`emes sur les groupes de substitutions.
\newblock \emph{Mathematische Annalen}, 5:\penalty0 584--594, 1872.

\bibitem[Varian(1974)]{equity1974}
H.~R. Varian.
\newblock Equity, envy, and efficiency.
\newblock \emph{Journal of Economic Theory}, 9:\penalty0 63--91, 1974.

\bibitem[Wielandt(2014)]{wielandt2014finite}
H.~Wielandt.
\newblock \emph{Finite permutation groups}.
\newblock Academic Press, 2014.

\bibitem[{\.Z}yczkowski and S{\l}omczy{\'n}ski(2014)]{zyczkowski2014square}
K.~{\.Z}yczkowski and W.~S{\l}omczy{\'n}ski.
\newblock Square root voting system, optimal threshold and $\pi$.
\newblock In \emph{Voting Power and Procedures}, pages 127--146. Springer,
  2014.

\end{thebibliography}

\end{document}